\documentclass[11pt,a4paper]{article}

\usepackage[a4paper,text={150mm,240mm},centering,headsep=10mm,footskip=15mm]{geometry}
\usepackage[T1]{fontenc}
\usepackage{graphicx}
\usepackage{epsfig, float}
\usepackage{pgf,tikz}
\usetikzlibrary{arrows}

\usepackage{amsmath,amssymb}
\usepackage{amsfonts}
\usepackage{url}
\usepackage{bbm}
\usepackage{mathrsfs} 
\usepackage[english]{babel}
\usepackage{hyperref}
\usepackage{slashed}

\newcommand{\R}{\mathbb{R}}
\newcommand{\C}{\mathbb{C}}
\newcommand{\N}{\mathbb{N}}

\newcommand{\id}{\mathbbm{1}}

\newcommand{\vx}{\mathbf{x}}

\newcommand{\be}{\begin{equation}}
\newcommand{\ee}{\end{equation}}
\newcommand{\psii}{\psi^{(1)}}
\newcommand{\psiii}{\psi^{(2)}}
\newcommand{\psitilde}{\widetilde{\psi}}
\newcommand{\config}{\mathcal{Q}}
\newcommand{\spacelike}{\mathscr{S}}
\newcommand{\coincidence}{\mathscr{C}}
\newcommand{\Banach}{\mathscr{B}}

\newtheorem{theorem}{Theorem}[section]
\newtheorem{lemma}[theorem]{Lemma}

\newenvironment{proof}[1][Proof:]{\begin{trivlist}
\item[\hskip \labelsep {\bfseries #1}]}{\end{trivlist}}

\newenvironment{definition}[1][Definition:]{\begin{trivlist}
\item[\hskip \labelsep {\bfseries #1}]}{\end{trivlist}}

\newcommand{\qed}{\hfill\ensuremath{\square}}

\title{Multi-time formulation of particle creation and annihilation via interior-boundary conditions}

\author{
Matthias Lienert\thanks{Fachbereich Mathematik, Eberhard-Karls-Universit\"at, Auf der Morgenstelle 10, 72076 T\"ubingen, Germany.
     E-mail: matthias.lienert@uni-tuebingen.de} \ and
Lukas Nickel\thanks{Mathematisches Institut, Ludwig-Maximilians-Universit\"at,
	Theresienstr.\ 39, 80333 M\"unchen, Germany. E-mail: nickel@math.lmu.de}
}

\date{July 31, 2019}

\begin{document}

\maketitle

\begin{abstract}
\noindent Interior-boundary conditions (IBCs) have been suggested as a possibility to circumvent the problem of ultraviolet divergences in quantum field theories. In the IBC approach, particle creation and annihilation is described with the help of linear conditions that relate the wave functions of two sectors of Fock space: $\psi^{(n)}(p)$ at an interior point $p$ and $\psi^{(n+m)}(q)$ at a boundary point $q$, typically a collision configuration.
 Here, we extend IBCs to the relativistic case. To do this, we make use of Dirac's concept of multi-time wave functions, i.e., wave functions $\psi(x_1,...,x_N)$ depending on $N$ space-time coordinates $x_i$ for $N$ particles. This provides the manifestly covariant particle-position representation that is required in the IBC approach. 
In order to obtain rigorous results, we construct a model for Dirac particles in 1+1 dimensions that can create or annihilate each other when they meet. Our main results are an existence and uniqueness theorem for that model, and the identification of a class of IBCs ensuring local probability conservation on all Cauchy surfaces.
Furthermore, we explain how these IBCs relate to the usual formulation with creation and annihilation operators.  The Lorentz invariance is discussed and it is found that apart from a constant matrix (which is required to transform in a certain way) the model is manifestly Lorentz invariant. This makes it clear that the IBC approach can be made compatible with relativity.\\[0.2cm]

\end{abstract}

\section{Introduction} \label{sec:intro}

The creation and annihilation of particles is a key feature of quantum field theory (QFT); however, it is usually connected with the problem of ultraviolet (UV) divergences. Interior-boundary conditions (IBCs) have been introduced as a possible way to circumvent the UV problem.
In the IBC approach (see \cite{ibc1,ibc2} for an introduction), one uses a wave function $\varphi$ on the configuration space $\config$ of a variable number of particles. Consider, for example, $\config = \config^{(1)} \cup \config^{(2)}$, where $\config^{(1)} = \R^d$ and $\config^{(2)} = \{ (\vx_1,\vx_2) \in (\R^d)^2 : \vx_1 \neq \vx_2\}$. The wave function can be represented as $\varphi = \big(\varphi^{(1)}, \varphi^{(2)}\big)$ where $\varphi^{(1)}$ is a single-particle wave function and $\varphi^{(2)}$ a two-particle wave function. An IBC then is a condition relating the value of $\varphi^{(2)}$ at a boundary point of $\config$, here $(\vx,\vx) \in \partial \config$, with the value of $\varphi^{(1)}$ at a suitable interior point, here at $\vx$.

The most important role of IBCs is to ensure that probability is transferred between the sectors of $n$ and $n+1$ particles while the total probability remains conserved. Such a transfer of probability corresponds to particle creation and annihilation. In the example, it describes
the process in which two particles can merge into a single one at a point $\vx$ or, conversely, a single particle at $\vx$ can split up into two. A certain notion of locality is important here: IBCs must only relate boundary points $q \in \partial \config$ with interior points $q' \in \mathring{\config}$ which canonically correspond to each other in the language of particle creation and annihilation. This means that $q = (\vx,\vx)$ must be related with $q' = \vx$, not with any other point.

Given this relation to particle creation and annihilation, it is not surprising that IBCs can also be derived from Hamiltonians which involve creation and annihilation operators. The creation operator can usually not be densely defined since it involves a delta function $\delta^{(d)}(\vx_1-\vx_2)$. 
 One can, however, obtain a boundary condition from the delta function by integrating over the Schr\"odinger equation, say in $\vx_2$ around the point $\vx_1$.
In this way, one arrives at a boundary condition that involves the wave function at points $(\vx_1,\vx_1) \in \partial \config$ and $\vx_1 \in \mathring{\config}$, i.e., at an IBC.

Previous works \cite{ibc1,ibc2,sltt:2017,stt:2018,KS:2016,ls:2018,lampart:2018,tumulka:2018} have focused on this relation of IBCs and non-relativistic (or pseudo-relativistic) Hamiltonians with particle creation and annihilation operators. It has been shown for simple models that understanding the creation part of the Hamiltonian as defining an IBC allows to make these models rigorous without the need for renormalization, which is usually required to treat the UV divergence. What is more, it is possible to explicitly state a well-defined version of the initial Hamiltonian, its domain being restricted by the IBC.

While these results seem promising, they have so far only been established for models non-relativistic or pseudo-relativistic dispersion relations, i.e., models with free Hamiltonians such as $-\Delta$ or $\sqrt{-\Delta +m^2}$. To have an impact on more realistic QFTs, it is important to extend them to a relativistic setting. Among other things, this requires using the Dirac operator to describe fermions. Besides, the idea of IBCs is intimately connected with the particle-position representation of the quantum state. As it is unclear how a wave function $\varphi^{(n)}(t,\vx_1,...,\vx_n)$ behaves under Lorentz transformations, a covariant alternative has to be utilized.

Fortunately, a relativistic version of particle-position representation has been revisited and developed significantly in recent years: the \textit{multi-time picture} of Dirac \cite{dirac:1932}, Fock and Podolsky \cite{dfp:1932}, Bloch \cite{bloch:1934}, Tomonaga \cite{tomonaga:1946} and Schwinger \cite{schwinger:1948}. We refer to \cite{LPT:2017} for a recent overview. In the multi-time picture, one considers a wave function $\psi^{(n)}(x_1,...,x_n)$ on \textit{spacetime con\-fi\-gu\-rations} $(x_1,...,x_n) \in (\R^{1+d})^n$ for $n$ particles. The relation to the usual single-time wave function $\varphi$ is straightforward:
\be
	\varphi^{(n)}(t,\vx_1,...,\vx_n) ~=~ \psi^{(n)}(t,\vx_1;...;t,\vx_n).
	\label{eq:singlemulti}
\ee
Under a Lorentz transformation $\Lambda $, the transformation behavior of $\psi$ is given by\footnote{This is the straightforward generalization to multiple times of the well-known product representation, say, for the $n$-particle Dirac equation $i \partial_t \varphi^{(n)} = \sum_{k=1}^n H^{\rm Dirac}_k \varphi^{(n)}$.}:
\be
	\psi^{(n)}(x_1,...,x_n) ~\stackrel{\Lambda}{\longmapsto}~ S[\Lambda]^{\otimes n} \psi^{(n)}(\Lambda^{-1} x_1,...,\Lambda^{-1} x_n),
	\label{eq:psitrafo}
\ee
where $S[\Lambda]$ are matrices forming a representation of the Lorentz group. In the context of QFT, $\psi$ can be represented as a sequence of $n$-particle wave functions, $\psi = (\psi^{(0)},\psi^{(1)},\psi^{(2)},...)$. The multi-time picture thus offers the covariant version of the particle-position representation that a relativistic treatment of IBCs requires.

Of course, several important questions remain, such as: 
\begin{enumerate}
	\item How can the idea of a probability flux between sectors of different particle numbers be formulated in the multi-time picture?
	\item How can IBCs be made compatible with the dynamics for a multi-time wave function?
\end{enumerate}
To clarify this complex of questions constitutes the main goal of the paper.  As IBCs concern the rigorous formulation of QFTs, it is crucial to prove the existence of the dynamics. For the multi-time case, this is challenging and, to the best of our knowledge, no general methods are available. (This should be contrasted with the single-time case where one can rely on a range of well-developed functional analytic tools centered around the self-adjointness of the Hamiltonian.) Therefore, a general answer to the above-mentioned problems seems presently out of reach. Our paper instead takes an exploratory approach. We identify the simplest nontrivial model which still shows the main features we are setting out to treat (Dirac operators, multi-time wave functions, particle creation and annihilation): a system of a variable number of at most $N$ indistinguishable Dirac particles in 1+1 spacetime dimensions. At this example we explain how to rigorously address the above-mentioned questions 1. and 2. about the multi-time formulation of IBCs. We expect that some of the resulting developments (especially those concerning the multi-time formulation) can be transferred to more general relativistic QFTs.

\paragraph{Structure of the paper.} We start by reviewing the basic concepts and results about multi-time wave functions which are relevant to our work (Sec.\ \ref{sec:multitimereview}). Next (Sec.\ \ref{sec:probcons}), we prove that a beautiful condition in terms of a differential form constructed from the multi-time wave function ensures (local) probability conservation on all Cauchy surfaces (Thm.\ \ref{thm:localcurrentcons}). In Sec.\ \ref{sec:2part}, we introduce our model. This is done for the case of $N=2$ sectors first, both for comprehensibility and because this case is used as a building block for the model with a general number $N \in \N$ of sectors of Fock space. We then identify a general class of IBCs which leads to local probability conservation (Thm.\ \ref{thm:2ibc}) and prove the existence and uniqueness of the model for $N=2$ (Thm.\ \ref{thm:2main}). 
Sec.\ \ref{sec:msec} deals with extending these results to general $N$. Our main result is the existence and uniqueness theorem \ref{thm:mmain}. We then discuss the relation of our model with a Hamiltonian with creation and annihilation operators (Sec.\ \ref{sec:mseccreat}). Lorentz invariance is briefly discussed in Sec.\ \ref{sec:li}. Section \ref{sec:proofs} contains the proofs of our theorems. We conclude with a discussion of the results, including an outlook on possible future directions (Sec.\ \ref{sec:discussion}).

\section{Multi-time wave functions} \label{sec:multitime}

We set $\hbar = 1 = c$ and use the spacetime metric $\eta = {\rm diag} (1,-1,...,-1)$ for (1+$d$)-dimensional Minkowski spacetime with $d$ spatial dimensions.

\subsection{Review of important concepts} \label{sec:multitimereview}

For a variable number of particles, a multi-time wave function becomes a so-called \textit{multi-time Fock function} (see \cite{pt:2013c} and also \cite{DV82b,DV85}). It can be represented as a sequence of $n$-particle wave functions $\psi^{(n)}$,
\be
	\psi = (\psi^{(0)},\psi^{(1)}, \psi^{(2)}, \psi^{(3)}, ... ).
\ee
Since the no-particle amplitude $\psi^{(0)}$ has neither time nor space arguments in the multi-time formalism, it is a constant complex number. Thus, we shall disregard $\psi^{(0)}$ for the rest of the paper and consider $\psi^{(n)}$ only for $n \in \N$, for which
\be
	\psi^{(n)} : \spacelike^{(n)} \subset (\R^{1+d})^n \rightarrow \C^{k(n)},~~~ (x_1,...,x_n) \mapsto \psi(x_1,...,x_n).
	\label{eq:multitimemap}
\ee
Here,  $k(n)$ denotes the number of spin components (e.g., $k(n) = 2^n$ for $n$ Dirac particles in $d=1$). The natural domain of $\psi^{(n)}$ is the set of spacelike configurations,
\be
	\spacelike^{(n)} = \{ (x_1,...,x_n) \in (\R^{1+d})^n : (x_i - x_j)^2 < 0 \ \forall i \neq j \}.
\ee
Here, $(x_i - x_j)^2 = (x_i^0-x_j^0)^2 - |\vx_i - \vx_j|^2$ denotes the Minkowski distance. 
The total configuration space is then given by the set
\be
	\spacelike = \bigcup_{n=1}^\infty \spacelike^{(n)} 
\ee
of spacelike configurations of a variable number of particles. Note that so far, the concept of a multi-time wave function is a straightforward relativistic extension of the usual single-time wave function on Fock space with configuration space $\config = \bigcup_{n=1}^\infty \R^{nd}_{\neq}$ where $\R^{nd}_{\neq} = \{ (\vx_1,...,\vx_n) \in \R^d : \vx_i \neq \vx_j \, \forall i \neq j)\}$.

The dynamics of $\psi$ is usually defined through a set of $n$ PDEs for each $\psi^{(n)}$, i.e.:
\be
	i \partial_{x_k^0} \psi^{(n)}(x_1,...,x_n) = (H_k \psi)^{(n)}(x_1,...,x_n),~~~k=1,...,n.
	\label{eq:multitimeeqs}
\ee
Here, the $H_k$ are differential expressions which may involve different sectors $\psi^{(m)}$. We can see their relation to the usual Hamiltonian as follows. Taking the time derivative of $\psi^{(n)}$ evaluated at equal times $t = x_1^0 = \cdots = x_n^0$ in a particular frame, we find that the single-time wave function $\varphi^{(n)}$ defined according to \eqref{eq:singlemulti} satisfies the Schr\"odinger equation
\be
	i \partial_t \varphi^{(n)}(t,\vx_1,...,\vx_n) = \left( \sum_{k=1}^n H_k \varphi \right)^{(n)}(t,\vx_1,...,\vx_n),
\ee
i.e., the Hamiltonian is given by $H = \sum_k H_k$.

It is understood that the multi-time equations \eqref{eq:multitimeeqs} can be rewritten in a manifestly covariant form, as in the example of free Dirac particles:
\be
	\big(i \gamma^\mu_k \partial_{x_k^\mu} - m \big) \psi^{(n)}(x_1,...,x_n) = 0,~~~k=1,...,n,
	\label{eq:freedirac}
\ee
where $\gamma^\mu_k$ denotes the $\mu$-th Dirac gamma matrix acting on the spin index of the $k$-th particle.

We emphasize that \eqref{eq:multitimeeqs} defines many equations for each $\psi^{(n)}$. In order for \eqref{eq:multitimeeqs} to have solutions for arbitrary initial data (e.g., data at all times equal to zero), the $H_k$ have to satisfy restrictive consistency conditions \cite{pt:2013a,ND:2016}. For the free equations \eqref{eq:freedirac}, these conditions are automatically satisfied. However, interaction potentials are typically inconsistent. Interaction via particle creation and annihilation, on the other hand, does yield a consistent mechanism of interaction for multi-time wave functions, at least on $\spacelike$ (see \cite{pt:2013c,pt:2013d}).

\subsection{Probability conservation for arbitrary Cauchy surfaces} \label{sec:probcons}

Multi-time wave functions carry a physical meaning as a probability amplitude for particle detection. It has recently been demonstrated \cite{LT:2017} that for a wide class of QFTs with local interactions and finite propagation speed, there is a \textit{Born rule for arbitrary Cauchy surfaces} $\Sigma \subset \R^{1+d}$. That means, a suitable quadratic expression $|\psi^{(n)}|_\Sigma^2$ in $\psi^{(n)}$ evaluated along $\Sigma$ yields the probability density to detect $n$ particles at locations $x_1,...,x_n\in \Sigma$. For example, for Dirac particles, one has:
\be
	|\psi^{(n)}|_\Sigma^2(x_1,...,x_n) ~=~ j^{\mu_1... \mu_n}(x_1,...,x_n) \, n_{\mu_1}(x_1) \cdots n_{\mu_n}(x_n),
	\label{eq:curvedborn}
\ee
where $n$ is the future-pointing normal vector field at $\Sigma$ and
\be
	j^{\mu_1... \mu_n}(x_1,...,x_n) ~=~ \overline{\psi}^{(n)} (x_1,...,x_n) \gamma_1^{\mu_1} \cdots \gamma_n^{\mu_n} \psi^{(n)} (x_1,...,x_n)
\label{eq:diraccurrent}
\ee
stands for the Dirac tensor current of the $n$-particle sector (see \cite[chap.\ 1.3]{lienert:2015c} for an explanation of tensor currents). For an equal-time surface $\Sigma_t$, we have $n=(1,0,...,0)$ and considering $(\gamma^0)^2 = \id$, we find $|\psi^{(n)}|_{\Sigma_t}^2= (\psi^{(n)})^\dagger \psi^{(n)}$, i.e., the usual $|\psi|^2$ probability density. Thus \eqref{eq:curvedborn} extends the usual Born rule in the appropriate geometric way.

\noindent Probability conservation then means
\be
	\sum_{n=1}^\infty \int_{\Sigma^n \cap \spacelike^{(n)}} d\sigma_1(x_1) \cdots d \sigma_n(x_n)~|\psi^{(n)}|_\Sigma^2(x_1,...,x_n) = 1 \quad \text{independently of} \ \Sigma.
\ee
To emphasize that we are dealing with a configuration space with a boundary, we have written $\Sigma^n \cap \spacelike^{(n)}$ for the domain of integration. This boundary $\partial \spacelike^{(n)}$ consists of the light-like configurations of $n$ particles. However, only a subset of $\partial \spacelike^{(n)}$ plays a role for probability conservation here, namely the set of coincidence points,
\be
	\coincidence^{(n)} = \{(x_1,...,x_n) \in (\R^{1+d})^n : \exists\, i \neq j : x_i = x_j \}.
\ee
In fact, for $d=1$ (the case we shall focus on later), the dimension of $\coincidence^{(n)}$ is large enough that probability can get lost through $\coincidence^{(n)}$. Accordingly, there must be conditions on the tensor currents $j^{\mu_1\cdots \mu_n}$ which ensure that the probability lost in this way gets redistributed to a different sector of Fock space. We shall now work out a suitable local condition which guarantees exactly that.

Before coming to the main result of the section, we introduce for every $n \in \N$ a certain $nd$-form, the \textit{current form} $\omega^{(n)}$ which is constructed from the tensor currents (see \cite{lienert:2015a,LN:2015} and \cite[chap. 1.3]{lienert:2015c}).
\begin{align}
	\omega^{(n)} = \sum_{\mu_1,...,\mu_n = 0}^d (-1)^{\mu_1 + \cdots + \mu_n} \,j^{\mu_1 ... \mu_n} \, dx_1^0 \wedge \cdots \widehat{dx_1}^{\mu_1}  \cdots \wedge dx_1^d\nonumber\\
\wedge \cdots \wedge dx_n^0 \wedge \cdots \widehat{dx_n}^{\mu_n}  \cdots \wedge dx_n^d,
\label{eq:currentform}
\end{align}
where $\widehat{(\cdot)}$ denotes omission.
Given $\omega^{(n)}$, the condition for probability conservation can be rewritten as follows.
\be
	\sum_{n=1}^\infty \int_{\Sigma^n\cap \spacelike^{(n)}} \omega^{(n)} = 1~~~{\rm independently~of~}\Sigma.
\label{eq:probcons0}
\ee

Now we specialize to $d=1$, denoting spacetime points by $x_i = (t_i,z_i)$. In $d=1$, the configuration space $\spacelike^{(n)}$ can be greatly simplified if one deals with a single species of indistinguishable particles (fermions). This is because we have
\be
	\spacelike^{(n)} = \bigcup_{\sigma \in S_n} \spacelike_\sigma^{(n)}, \quad 	\spacelike_\sigma^{(n)} = \{ (t_1,z_1,...,t_n,z_n) \in \spacelike : z_{\sigma(1)} < \cdots < z_{\sigma(n)} \}.
\label{eq:spacelikesigma}
\ee
where $S_n$ denotes the permutation group.
Thus, it is sufficient to consider $\spacelike_1^{(n)}= \spacelike_{\rm id}^{(n)}$ as the configuration space of $n$ indistinguishable particles in $d=1$, and $\spacelike_1 = \bigcup_{n=1}^\infty \spacelike_1^{(n)}$ for a variable particle number. The idea is to only formulate the multi-time equations \eqref{eq:multitimeeqs} and the IBCs on $\spacelike_1$ (and its boundary). Once a solution $\psi$ on $\spacelike_1$ is found, one can obtain an appropriately normalized multi-time wave function $\widetilde{\psi}$ on $\spacelike$ by anti-symmetric extension and normalization of $\psi^{(n)}$ with a factor $\frac{1}{\sqrt{n!}}$. On each $\spacelike^{(n)}_\sigma$:
\be \label{eq:renormalizedwf}
	\widetilde{\psi}^{(n)}_{s_1 ... s_n}(x_1,...,x_n) = \frac{\mathrm{sgn}( \sigma)}{\sqrt{n!}} \, \psi^{(n)}_{s_{\sigma(1)} ... s_{\sigma(n)}}(x_{\sigma(1)},...,x_{\sigma(n)}),
\ee
where $s_k$ is the spin index of the $k$-th particle.

We are now prepared to prove the condition for probability conservation. For technical reasons, we introduce a highest possible number $N \in \N$ of particles in the system.

\begin{theorem}[Condition for local probability conservation.] \label{thm:localcurrentcons}
Let $N \in \N$ and let $j = (j^{\mu}, j^{\nu \rho}, \cdots, j^{\mu_1...\mu_N})$ such that each $j^{\mu_1 ... \mu_n}$ is differentiable on $\spacelike_1^{(n)}$ and continuous on $\overline{\spacelike}_1^{(n)}$. Furthermore, let each $j^{\mu_1...\mu_n}$ be compactly supported on all sets of the form $\Sigma^n\cap \spacelike_1^{(n)}$ where $\Sigma \subset \R^2$ is a smooth Cauchy surface. Let $\omega^{(n)}$ denote the $n$-form given by $j^{\mu_1...\mu_n}$ as in \eqref{eq:currentform}. Then:
\begin{enumerate}
	\item Global probability conservation in the sense of 
	\be
			\sum_{n=1}^N \int_{\Sigma^n \cap \spacelike^{(n)}_1} \omega^{(n)} = 1~~~\text{for all Cauchy surfaces } \Sigma
\label{eq:probcons}
 \ee
is ensured by the following condition \emph{(local probability conservation)}:
	\be
		\left\{ \begin{array}{l} {\rm d} \omega^{(N)} = 0,\\
		{\rm d} \omega^{(n)} = \sum_{k=1}^n {\Phi_k}^*\omega^{(n+1)} , ~~ n=1,...,N-1.
		\end{array}\right.
		\label{eq:localcurrentcons}
	\ee
	Here, $\Phi_k$ is defined by $(k=1,...,n)$:
\begin{align}
	\Phi_k : ~&\overline{\mathscr{S}}_1^{(n)} \rightarrow \mathscr{C}^{(n+1)}_k = \{ (x_1,...,x_{n+1}) \in \partial \mathscr{S}^{(n+1)}_1 : x_k = x_{k+1}\},\nonumber\\
& (x_1,...,x_k,...,x_{n+1}) \mapsto (x_1,...,x_k,x_k,...,x_{n+1}),
\label{eq:phik}
\end{align}
and $\big({\Phi_k}^*\omega^{(n+1)}\big)(\cdot) = \omega^{(n+1)}(\Phi_k(\cdot))$ denotes the pullback of $\omega^{(n+1)}$ by $\Phi_k$. Moreover, evaluation of $\omega^{(n+1)}$ along $\mathscr{C}^{(n+1)}_k \subset \partial \mathscr{S}^{(n+1)}_1$ refers to the limit of $\omega^{(n+1)}(q)$ for $q \rightarrow \partial \mathscr{S}^{(n+1)}_1$ in $\mathscr{S}^{(n+1)}_1$.
	\item Let $\varepsilon_{\mu \nu}$ denote the Levi-Civita symbol. In terms of the tensor currents $j^{\mu_1 ... \mu_n}$, \eqref{eq:localcurrentcons} is then equivalent to:
	\be
	\left\{ \begin{array}{l} \partial_{x_k^{\mu_k}} j^{\mu_1... \mu_k... \mu_N}= 0 \quad \forall \, k=1,...,N~\text{on} ~ \spacelike_1^{(N)}, \\
		\varepsilon_{\rho \sigma}\, j^{\mu_1 ... \mu_{k-1} \, \rho \, \sigma \, \mu_{k+1} ... \mu_n}(x_1,...,x_k,x_k,...,x_n) = (-1)^k \partial_{x_k^{\mu_k}} j^{\mu_1... \mu_k ... \mu_n}(x_1,...,x_n)\\
\forall \, n=1,...,N-1,~ k=1,...,n, ~\forall \, (x_1,...,x_n) \in \mathscr{S}_1^{(n)}.
		\end{array}\right.
		\label{eq:currentcondition}
	\ee
\end{enumerate} 
\end{theorem}

The proof can be found in Sec. \ref{sec:prooflocalcurrentcons}.

\paragraph{Remark.} It is remarkable that probability conservation is ensured by the beautiful geometric condition \eqref{eq:localcurrentcons}. We call this condition \textit{local probability conservation.} The adjective ``local'' deserves some explanation. One can see from the proof that we demand a certain detailed balance between the probability flux of the ($n$+1)-particle sector into the set where two of the $n$+1 points coincide and the global influx into the $n$-particle sector. We then mean by ``local current conservation'' that this redistribution of probability happens only between configurations $(x_1,...,x_k, x_k,...,x_n) \in \partial \spacelike^{(n+1)}$ and $(x_1,...,x_k,...,x_n) \in \spacelike^{(n)}$, i.e., configurations which correspond to each other canonically in the particle-position representation: $(x_1,...,x_k, x_k,...,x_n)$ refers to a configuration of $n$+1 particles on a certain Cauchy surface where two of the particles meet, and $(x_1,...,x_k,...,x_n)$ is the configuration where just $n$ particles are present but at the same locations. So the straightforward interpretation of what happens here is that two of the particles merge to form a single one, i.e., one of the particles gets annihilated. It becomes evident that \eqref{eq:localcurrentcons} regulates the way particle creation and annihilation can happen. In more sophisticated theories with exchange particles, we expect that an analogous condition will hold on configurations where an exchange particle reaches one of the other particles. 

In the following, we shall define IBCs as linear relations between the spin components of $\psi^{(n+1)}(x_1,...,x_k,x_k,...,x_n)$ and $\psi^{(n)}(x_1,...,x_k,...,x_n)$ (for $n \geq 1$) which ensure local probability conservation \eqref{eq:localcurrentcons}. We shall identify a large class of local IBCs for the simplest case of $N=2$ sectors of Fock space (i.e., a model including the 1 and 2 particle sectors). The results will serve as a building block for the case of a general $N$.

\section{A building block: the case $N=2$} \label{sec:2part}

\subsection{The model} \label{sec:2secmodel}

We consider a variable number of at most $N=2$ indistinguishable Dirac particles in one spatial dimension. This is the simplest non-trivial case where particle creation and annihilation is possible. We have $\psi = (\psi^{(1)}, \psi^{(2)})$ where $\psi^{(1)}$ and $\psi^{(2)}$ are maps of the form \eqref{eq:multitimemap} with $k(n) = 2^n$ spin components on the reduced configuration spaces $\spacelike_1^{(n)}$, see \eqref{eq:spacelikesigma}. Explicitly, we write:
\begin{equation} \begin{split}
\psii: \R^2 \rightarrow \C^2,~~~ \psii(t,z) &=  \left( \begin{array}{c}
\psii_-(t,z) \\ \psii_+(t,z)
\end{array} \right) , \\
\psiii: \spacelike_1^{(2)} \rightarrow \C^4,~~~ \psiii(t_1,z_1,t_2,z_2) & ~=~ \left( \begin{array}{c}
\psiii_{--}(t_1,z_1,t_2,z_2) \\ \psiii_{-+}(t_1,z_1,t_2,z_2)  \\ \psiii_{+-}(t_1,z_1,t_2,z_2) \\ \psiii_{++}(t_1,z_1,t_2,z_2) 
\end{array} \right). \end{split}
\end{equation}
The dynamics is defined as follows. $\psiii$ obeys the free multi-time Dirac equations on $\spacelike_1^{(2)}$ (here in Hamiltonian form):
\be \label{eq:2modelupper}
	i \partial_{t_k} \psiii(t_1,z_1,t_2,z_2) = H_k^{\rm Dirac} \psiii(t_1,z_1,t_2,z_2) ,~~k=1,2,
\ee
where $H_k^{\rm Dirac} = -i \gamma^0_k \gamma^1_k \partial_{z_k} + m \gamma^0_k$ is the Dirac Hamiltonian acting on the variables of the $k$-th particle.  $\psii$ evolves according to:
\begin{equation}  \label{eq:2modellower}
	i\partial_t \psii(t,z) = H^{\rm Dirac} \psii(t,z) - A \psiii(t,z,t,z),
\end{equation}
where $A$ is a $2\times 4$ matrix. The term $A \psiii(t,z,t,z)$ creates a coupling between the two sectors and allows for a global gain/loss of probability in the 1-particle sector. The matrix $A$ is constrained by current conservation as will be explained in Sec.\ \ref{sec:2secibc}.

Furthermore, $\psii$ and $\psiii$ need to obey the following IBC:
\begin{equation} \label{eq:2modelIBC}
\psiii_{-+}(t,z,t,z)- e^{i \theta} \psiii_{+-}(t,z,t,z) = B \psii(t,z),
\end{equation}
for some $\theta \in [0,2\pi)$. $B$ is a $1\times 2$ matrix which can be expressed in terms of $A$ (see Sec. \ref{sec:2secibc}). Expressions involving $\psi$ on boundary points $q\in \partial \spacelike$, such as $\psiii_{+-}(t,z,t,z)$, denote limits of $\psi$ within $\spacelike_1$ towards the boundary, e.g., $\psiii_{+-}(t,z,t,z) = \lim_{\varepsilon \rightarrow 0} \psiii_{+-}(t,z-\varepsilon,t,z+\varepsilon)$. These limits can be understood in the literal sense; we will consider only continuously differentiable and bounded wave functions in this paper.

The form of the IBC \eqref{eq:2modelIBC} can be motivated as follows.\footnote{We have first learned about this possibility from Roderich Tumulka (private communication). After the completion of the present paper, a general (single-time) formulation of IBCs for codimension-1 boundaries and multi-particle Dirac Hamiltonians has appeared in \cite{stt:2018}.} In the case of no coupling between the two sectors (i.e., $A=0$, $B=0$), the model corresponds to free motion for the 1-particle sector and pure delta interactions for the 2-particle sector. Such relativistic delta interactions for multi-time wave functions have been treated in \cite{lienert:2015a} and the appropriate form of the boundary conditions is known from there. If a coupling between the 1-particle and 2-particle sectors is desired, it is natural to include a linear term $B \psii(t,z)$ on the right hand side. Moreover, because a transfer of probability between the two sectors is expected, we need to add a source term to the free Dirac equation for the 1-particle sector. This source term should be linear and can only depend on $(t,z)$. This suggests that it should have the form $A \psiii(t,z,t,z)$. (A similar approach has been used in \cite[Sec. 2.3]{ibc2} to introduce IBCs in a non-relativistic context.)

In order to simplify the problem, we choose the representation $\gamma^0 = \sigma_1$ and $\gamma^1 = \sigma_1 \sigma_3$ where $\sigma_i,~i=1,2,3$ denote the Pauli matrices. If we then set $m=0$, the Dirac Hamiltonian becomes diagonal:
\begin{equation} \label{eq:formoffreeDirac}
H^{\rm Dirac} = \left( \begin{array}{cc} 
-i\partial_z & 0 \\ 0 & i\partial_z
\end{array} \right).
\end{equation}
This diagonal form makes it possible to use a generalized method of characteristics which has been developed in \cite{lienert:2015a,LN:2015}. For this reason we shall focus on the massless case first (the massive case is treated in Sec.\ \ref{sec:massive} for $N$ sectors of Fock space).

Our goal is to prove the existence and uniqueness of solutions of the system of equations \eqref{eq:2modelupper}, \eqref{eq:2modellower} and \eqref{eq:2modelIBC} with $m=0$ (the case $m>0$ will be treated in Sec.\ \ref{sec:massive}) and initial values given by
\begin{equation} \label{eq:2modelinitial}
\psii|_{t=0} = \psii_0 \in C_b^1(\R,\C^2), \quad \psiii|_{t_1=t_2=0} = \psiii_0 \in C_b^1( \{ (z_1,z_2) \in \R^2, z_1<z_2 \}, \C^4).
\end{equation}
Here, $C^1_b$ denotes the space of continuously differentiable functions which are bounded and have bounded derivatives. In addition, we require the following compatibility conditions between initial values and the IBC:
\begin{align}\label{eq:IBCforinitialvalues}
& \psi^{(2)}_{0,-+}(z,z) - e^{i\theta} \psi^{(2)}_{0,+-}(z,z) = B \psi_{0}^{(1)}(z)
\\ & B \left( H^{\rm Dirac} \psii_0(z) - A\psiii_0(z,z) \right) = i \left( \partial_{z_2} - \partial_{z_1} \right) \left. \left( \psiii_{0,-+}(z_1,z_2)+e^{i\theta}\psiii_{0,+-}(z_1,z_2)  \right) \right|_{z_1=z_2=z} \label{eq:IBCforinitialvalues2}
\end{align}
The first condition expresses that the initial data must satisfy the IBC at time $t=0$. The second condition turns out to be necessary to obtain a $C^1$ solution (see Sec.\ \ref{sec:proof2main} to understand fully why \eqref{eq:IBCforinitialvalues2} arises).
These two conditions can be satisfied as follows. Choose $\psii_0 \in C_b^1$ at will. Then \eqref{eq:2modelinitial} can be read as a boundary condition for $\psiii_0$ which is easy to fulfill. The remaining condition \eqref{eq:IBCforinitialvalues2} is a boundary condition for the $(\partial_{z_2} - \partial_{z_1})$-derivative of $\psiii_0$. It can be satisfied easily as well, as the derivatives of $\psiii_0$ at a boundary point can be chosen independently of the value of $\psiii_0$ at that point.

\begin{theorem} \label{thm:2main}
Let $T>0$ and $A \in \C^{2 \times 4}, B \in \C^{1 \times 2}$ be arbitrary. Then for  $t_1, t_2 \in [-T,T]$ there exists a unique $C^1_b$-solution $\psi$ of the initial boundary value problem \eqref{eq:2modelupper}, \eqref{eq:2modellower}, \eqref{eq:2modelIBC} with given initial values as in \eqref{eq:2modelinitial} that satisfy \eqref{eq:IBCforinitialvalues} and \eqref{eq:IBCforinitialvalues2}. We call such a $\psi$ a \emph{global solution}.
\end{theorem}
We shall now determine which matrices $A$ and $B$ lead to local probability conservation. (The existence and uniqueness theorem holds for arbitrary constant matrices $A, B$.) After that, we compare the form of the equations and IBCs with a more familiar Hamiltonian involving creation and annihilation operators and establish a relation between the two formulations (Sec.\ \ref{sec:2secibc}). The proof of Thm.\ \ref{thm:2main} is given is Sec.\ \ref{sec:proof2main}.

\subsection{Probability conservation and IBCs} \label{sec:2secibc}

We need to check which matrices $A$ and $B$ ensure the condition \eqref{eq:localcurrentcons} (or equivalently \eqref{eq:currentcondition})  for local probability conservation. Eq. \eqref{eq:currentcondition} yields the following two conditions. For $n=2$:
\be
	\partial_{x_1^\mu} j^{\mu \nu}(x_1,x_2) = 0 = \partial_{x_2^\nu} j^{\mu \nu}(x_1,x_2)~~\text{on } \spacelike_1^{(2)},
	\label{eq:zerodiv}
\ee
and for $n=1$:
\be
	(j^{01} - j^{10})(x,x) = - \partial_\mu j^\mu(x)~~\forall \, x \in \R^2,
	\label{eq:balance2}
\ee
where $j^{\mu \nu} = \overline{\psi}^{(2)} \gamma_1^\mu \gamma_2^\nu \psiii$ and $j^\mu = \overline{\psi}^{(1)} \gamma^\mu \psii$.

Now, \eqref{eq:zerodiv} is already ensured by the free multi-time Dirac equations \eqref{eq:2modelupper} for the two-particle sector. (This can be verified easily using \eqref{eq:2modelupper} and its adjoint equation.) It will be the role of the IBC \eqref{eq:2modelIBC}  to ensure \eqref{eq:balance2}. We now calculate both sides of \eqref{eq:balance2} in detail to see which relation the matrices $A$ and $B$ need to satisfy. On the one hand, we have:
\be
	\partial_\mu j^\mu(x) \stackrel{\eqref{eq:2modellower}}{=} - 2 \Im \left( {\psii}^\dagger(x) \, A \psiii(x,x) \right).
\ee
On the other hand, the two-particle flow out of the set of coincidence points is given by \cite{lienert:2015a}:
\be
	\big(j^{01}-j^{10}\big)(x,x) = 2 \big( |\psiii_{+-}|^2 - |\psiii_{-+}|^2 \big)(x,x).
\ee
Thus, condition \eqref{eq:balance2} becomes:
\be
	\Im \left( {\psii}^\dagger(x) \, A\,  \psiii(x,x) \right) = \big( |\psiii_{+-}|^2 - |\psiii_{-+}|^2 \big)(x,x).
	\label{eq:2currentbalance}
\ee
It is the content of the following theorem to identify a general class of IBCs which ensure this condition.

\begin{theorem} \label{thm:2ibc}
	The most general translation invariant class of IBCs of the form \eqref{eq:2modelIBC} which ensures \eqref{eq:2currentbalance} (and hence \eqref{eq:currentcondition}) for the model \eqref{eq:2modelupper}, \eqref{eq:2modellower} is given by a phase $\theta \in [0,2\pi)$ and constant matrices $A, B$ defined as follows.
	\be
		A^\dagger = \left( \begin{array}{c} 0~0\\ \, \widetilde{A} \\ 0~0\end{array}\right),
	\label{eq:N}
	\ee
 	where $\widetilde{A}$ is a complex $2 \times 2$ matrix of the form
	\be
		\widetilde{A} = \left( \begin{array}{cc} w_1 & w_2\\ w_1 e^{i\phi} & w_2 e^{i \phi}\end{array}\right)
		\label{eq:ntilde}
	\ee
	with $w_1,w_2 \in \C$ and $\phi \in [0,2\pi)$.\\
	Furthermore, $B$ can be expressed completely in terms of $\widetilde{A}$ as
	\be
		B = \frac{1}{2i} \, (1,e^{i \theta}) \widetilde{A}.
		\label{eq:M}
	\ee
 \end{theorem}
The proof is given in Sec. \ref{sec:proof2ibc}.

\paragraph{Remarks.}
\begin{enumerate}
	\item The IBC \eqref{eq:2modelIBC} describes the interaction effect of the annihilation of two particles into one if they meet (and conversely the creation of two particles out of one). It seems reasonable that the interaction between any two particles should be of the same form, regardless of which two particles $k, k+1$ meet, of the particle number of the sectors that are considered and of the total number of sectors. Using this principle, the form of the matrices $A$ and $B$ from theorem \ref{thm:2ibc} will be a crucial building block for a model with $N$ sectors.
\item \textit{Spin index notation.} It is helpful to write the matrices $A$ and $B$ using spin index notation. We have:
\be
	B = B^s~~~ \text{and}~~~A = A_s^{t u},
\ee
where $s,t,u = \pm 1$. An upper index means that the respective matrix will be contracted with a respective lower spin index of $\psi$. A lower index indicates that the matrix times $\psi$ will have that spin index in addition to the spin indices which do not get summed over.
\end{enumerate}

\section{The case of $N>2$ sectors of Fock space} \label{sec:msec}

For simplicity, we treat the massless case first as the proof is more direct and transparent in this case. The case $m>0$ is addressed in Section \ref{sec:massive}.

\subsection{The massless case} \label{sec:massless}

We now generalize both the dynamics as well as the existence and uniqueness results to the case of $N>2$ sectors of Fock space. The wave function then has the form $\psi = \left( \psi^{(1)},...,\psi^{(N)} \right)$, where each $\psi^{(n)}$ is a map
\begin{equation}
\psi^{(n)}: \mathscr{S}_1^{(n)} \to \C^{2^n}, ~~~(t_1,z_1;...;t_n,z_n) \mapsto  \left( \begin{array}{c}
\psi^{(n)}_{-...--}(t_1,z_1;...;t_n,z_n) \\ \psi^{(n)}_{-...-+}(t_1,z_1;...;t_n,z_n)  \\ \psi^{(n)}_{-...+-}(t_1,z_1;...;t_n,z_n)  \\ \vdots \\\psi^{(n)}_{+...++}(t_1,z_1;...;t_n,z_n) 
\end{array} \right).
\end{equation}
For readability, we shall sometimes use semicolons to divide space-time arguments associated with different particle indices.  
As evolution equations, we consider multi-time Dirac equations on $\mathscr{S}_1^{(n)}$ (here in Hamiltonian form) 
\be \begin{split}
i \partial_{t_k} \psi^{(n)} = H^{\rm Dirac}_k \psi^{(n)} + f^{(n)}_k	 ~~~ \Longleftrightarrow ~~~	i(\partial_{t_k} - s_k \partial_{z_k}) \psi^{(n)}_{s_1...s_n} = f^{(n)}_{k,s_1...s_n},
		\\ n=1,...,N;~ k=1,...,n;~s_1,...,s_n = \pm 1.
		\label{eq:multitimeeqwithsources} \end{split}
	\ee
According to the remark at the end of the previous section, the source terms are given by:
	\be
		f_{k,s_1...s_n}^{(n)}(x_1,...,x_n) = (-1)^k \sum_{t,u = \pm 1}A_{s_k}^{tu} \psi^{(n+1)}_{s_1...s_{k-1}\,  t \, u \, s_{k+1}...s_n}(x_1,...,x_k,x_k,x_{k+1},...,x_n)
	\label{eq:sourceterms}
	\ee
	for $n=1,...,N-1$, and
	\be
 		f_{k,s_1...s_N}^{(N)} \equiv 0~~\forall k.
 		\label{eq:mthsourceterm}
 	\ee
In addition, there is an interior-boundary condition. For some $\theta \in [0,2\pi)$, 
\begin{align}
		&\left( \psi^{(n+1)}_{s_1...s_{k-1} -+ s_{k+1} ... s_n}- e^{i\theta}\psi^{(n+1)}_{s_1...s_{k-1} +- s_{k+1} ... s_n} \right) (x_1,...,x_k,x_k, x_{k+1},...,x_n)\nonumber\\
	 &=~ \sum_s B^s \psi^{(n)}_{s_1 ... s_{k-1} s \, s_{k+1} ... s_{n}}(x_1,...,x_n)
	 \label{eq:mgeneralibcs}
	\end{align}
	for all $n=1,...,N-1$, all $k = 1,...,n$, all spin components $s_1,...,s_{k-1},s_{k+1},...,s_n = \pm 1$ and all $x_1,...,x_n \in \overline{\mathscr{S}}_1^{(n)}$. $A, B$ are the same matrices as for $N=2$ (see Thm. \ref{thm:2ibc}).

Initial data are given by
	\begin{equation} \label{eq:minitial}
	\psi^{(n)}|_{t_1=...=t_n=0}=\psi_0^{(n)} \in C^1_b( Z_n, \C^{2^n}),
	\end{equation}
 where
\begin{equation} \label{eq:Zn}
Z_n := \{ (z_1,...,z_n) \in \R^n| z_1<...<z_n  \}
\end{equation}	
and, as before, $C^1_b$ denotes the set of continuously differentiable functions which are bounded and have bounded partial derivatives. 

The initial data have to satisfy the following compatibility conditions (which are the analogs of Eqs.\ \eqref{eq:IBCforinitialvalues} and \eqref{eq:IBCforinitialvalues2} for general $N$).
\begin{align}
	&\left( \psi^{(n+1)}_{0,s_1...s_{k-1} -+ s_{k+1} ... s_n}- e^{i\theta}\psi^{(n+1)}_{0,s_1...s_{k-1} +- s_{k+1} ... s_n} \right) (z_1,...,z_k,z_k, z_{k+1},...,z_n)
	\nonumber\\ \label{eq:IBCforinitialvaluesN}
	 &=~ \sum_s B^s \psi^{(n)}_{0,s_1 ... s_{k-1} s \, s_{k+1} ... s_{n}}(z_1,...,z_n), \\ 
&	\sum_s  B^s \left( H^{\rm Dirac}_k \psi_{s_1...s_{k-1}ss_{k+1}...s_n}^{(n)} + f^{(n)}_{s_1...s_{k-1}ss_{k+1}...s_n} \right)(0,z_1;...0,z_{k-1};0,z;0,z_{k+2}...;0,z_{n+1})    \nonumber
\\ &  =i  \left( \partial_{z_{k}} - \partial_{z_{k+1}}\right)  \left.\left( \psi^{(n+1)}_{0,s_1...s_{k-1} -+ s_{k+1} ... s_n} + e^{i \theta} \psi^{(n+1)}_{0,s_1...s_{k-1} +- s_{k+1} ... s_n} \right)(z_1,...,z_{n+1}) \right|_{z_k=z_{k+1}=z} \label{eq:IBCforinitialvaluesN2}
\end{align}
The first condition ensures compatibility of initial data and IBC; the second condition is needed to obtain a $C^1$-solution.

We consider the following function spaces for $\psi^{(n)}$:
\begin{equation} \label{eq:banachn}
\Banach_n := C_b ^1 ( \{(t_1,z_1;...;t_n,z_n) \in \mathscr{S}^{(n)}_1 : t_k \in [0,T], (z_1,...,z_n) \in Z_n\}, \C^{2^n}),
\end{equation}
In contrast to $N=2$, we only admit positive times to avoid technical complications.

Accordingly, $\psi$ is an element of 
\be
	\Banach ~:=~ \bigoplus_{n=1}^N \Banach_n.
\ee

Our main results are the following theorems, the first one about the existence and uniqueness of solutions and the second one about probability conservation.

\begin{theorem}[Existence and uniqueness of solutions.] \label{thm:mmain}
Let $T>0$. Then for all initial data given by \eqref{eq:minitial} with \eqref{eq:IBCforinitialvaluesN} and \eqref{eq:IBCforinitialvaluesN2} and for all $0 \leq t_1,...,t_N \leq T$,  there exists a unique solution $\psi \in \Banach$ of the initial boundary value problem \eqref{eq:multitimeeqwithsources}--\eqref{eq:mgeneralibcs}.
\end{theorem}
The proof is given in Sec. \ref{sec:proofmmain}.

\begin{theorem}[Local probability conservation.] \label{thm:nibc}
	Let $\theta \in [0,2\pi)$ and let $A$ and $B$ be the matrices from Thm. \ref{thm:2ibc}. Then the IBCs \eqref{eq:mgeneralibcs} ensure local probability conservation in the sense of \eqref{eq:localcurrentcons}.
\end{theorem}
The proof can be found in Sec. \ref{sec:proofnibc}.

Next, we establish a relation between the equations of our model and the usual formulation of QFTs via creation and annihilation operators.

\subsection{The massive case $m \neq 0$} \label{sec:massive}

Theorem \ref{thm:mmain} is proven by a fixed point argument. A generalization of this argument can also be applied to yield the same statement in the case of a non-zero mass $m > 0$.\footnote{We are grateful to an anonymous referee for pointing this out to us.} Probability conservation (Thm. \ref{thm:nibc}) trivially remains true as the mass term does not influence functional form of the tensor currents $j^{\mu_1...\mu_N}$.

\begin{theorem}[Existence and uniqueness in the massive case.] \label{thm:massive}
Let $m \in \R$ and $T>0$. Furthermore, let initial data given be given by \eqref{eq:minitial} with \eqref{eq:IBCforinitialvaluesN} and \eqref{eq:IBCforinitialvaluesN2}. Then, for all $0 \leq t_1,...,t_N \leq T$, there exists a unique solution $\psi \in \mathscr{B}$ of the initial boundary value problem \eqref{eq:multitimeeqwithsources},\eqref{eq:mgeneralibcs} with source terms 
	\be \begin{split}
		f_{k,s_1...s_n}^{(n)}(x_1,...,x_n) ~=~ & (-1)^k \sum_{t,u = \pm 1}A_{s_k}^{tu} \psi^{(n+1)}_{s_1...s_{k-1}\,  t \, u \, s_{k+1}...s_n}(x_1,...,x_k,x_k,x_{k+1},...,x_n) \\ & + m \left(\gamma^0_k \psi^{(n)}\right)_{s_1...s_n}
	\end{split} \label{eq:sourcetermswithmass}
	\ee
	for $n=1,...,N-1$, and
	\be
 		f_{k}^{(N)} = m \gamma^0_k \psi^{(N)}.
 		\label{eq:mthsourcetermwithmass}
 	\ee
\end{theorem}

The proof is given in Sec.\ \ref{sec:proofmassive}.

Next, we establish a relation between the equations of our model and the usual formulation of QFTs via creation and annihilation operators.

\section{Relation to creation/annihilation operators} \label{sec:mseccreat}

Usually, one introduces QFTs using creation and annihilation operators. In the previous sections, we have chosen a different way. It is, therefore, important to connect the two approaches. In order to do this, we now consider the single-time version of our model, so that the multi-time wave function $\psi$ reduces to the single-time wave function $\varphi$ in the Schr\"odinger picture of QFT. To obtain $\varphi$ from $\psi$ when $\psi$ is defined only on $\mathscr{S}_1$, we need to combine \eqref{eq:renormalizedwf} and \eqref{eq:singlemulti}. For $(t,z_1,...,t,z_N) \in \mathscr{S}_\sigma^{(n)}$, we have:
\begin{equation} \label{eq:comparephipsi}
\varphi^{(n)}_{s_1 ... s_n}(t;z_1,...,z_n) = \frac{\mathrm{sgn}( \sigma)}{\sqrt{n!}} \, \psi^{(n)}_{s_{\sigma(1)} ... s_{\sigma(n)}}(t, z_{\sigma(1)},...,t, z_{\sigma(n)}).
\end{equation}
This allows us to identify the interaction part of the Hamiltonian for the single-time version of our model. The annihilation terms in the $n$-th sector are obtained as the sum over the source terms in \eqref{eq:sourceterms}, i.e.\
\begin{align}
	&\left(H_{\rm int}^{\rm ann} \varphi \right)_{s_1...s_n}^{(n)}(t;z_1,...,z_n)\nonumber\\
&= \sqrt{n+1} \sum_{k=1}^n \sum_{t,u=\pm1} (-1)^k  A^{tu}_{s_k} \varphi^{(n+1)}_{s_1...s_{k-1}tus_{k+1}...s_n}(t;z_1,...,z_k,z_k,z_{k+1},...,z_n),
\end{align}
for $1 \leq n \leq N-1$. The factor of $\sqrt{n+1}$ comes from \eqref{eq:comparephipsi}. 

In the particle-position representation, the creation and annihilation operators are given as \cite{schweber:1961}
\begin{align}
	(a_r(z) \varphi)^{(n)}_{s_1...s_n}(t; z_1,...,z_n) &=  \, \sqrt{n+1} \ \varphi^{(n+1)}_{r \, s_1...s_n}(t;z,z_1,...,z_n), \label{eq:creator}\\
	(a_r^\dagger(z) \varphi)^{(n)}_{s_1...s_n}(t; z_1,...,z_n) &= \frac{1}{\sqrt{n}} \sum_{k=1}^{n} (-1)^{k+1} \delta_{s_k \, r} \delta(z-z_k) \varphi^{(n-1)}_{s_1...\widehat{s}_k ... s_n}(t;z_1,..., \widehat{z}_k,...,z_n). \label{eq:annihilator}
\end{align}
This allows us to rewrite $H_{\rm int}^{\rm ann}$ as follows:
\be \label{eq:annihilationint}
	\left( H_{\rm int}^{\rm ann} \varphi \right)^{(n)} =  \left( \int d z \sum_{rst}  A_r^{st} a^\dagger_r(z)  a_s(z) a_t(z) \varphi \right)^{(n)}, \quad 1 \leq n \leq N-1.
\ee

Now, the creation part $H_{\rm int}^{\rm cre}$ of the Hamiltonian cannot be read off straightforwardly from \eqref{eq:sourceterms}. (As we will see, the reason is that it is not a well-defined quantity as it contains $\delta$-functions.) We can, however, obtain $H_{\rm int}^{\rm cre}$ by taking the adjoint of $H_{\rm int}^{\rm ann}$:
\be
	H_{\rm int}^{\rm cre} = (H_{\rm int}^{\rm ann})^\dagger = \int dz \sum_{rst} (A_r^{st})^* a^\dagger_t(z) a^\dagger_s(z) a_r(z).
\ee
The action of $H_{\rm int}^{\rm cre}$ on wave functions is 
\be \label{eq:trilala} \begin{split}
  ( H^{\rm cre}_{\rm int} \varphi)^{(n)}_{s_1...s_n}(t; z_1,...,z_n) = \sum_{\substack{ j,k = 1 \\ j \neq k}}^n \sum_r (-1)^{j+1} \frac{1}{\sqrt{n}} (A^{s_k s_j}_r)^*
   \delta(z_j-z_k) \times \\ \times \varphi^{(n-1)}_{s_1...\widehat{s_j}...r...s_n}(t;z_1,...,\widehat{z_j},...,z_k,...,z_n). \end{split}
\ee
This is indeed not well-defined because the $\delta$-distribution is not an element of $L^2$.

We shall now show that our model of Sec.\ \ref{sec:msec} is a rigorous (and multi-time) version of the single-time model with Hamiltonian $H = H^{\rm free} + H_{\rm int}$. Here, $H^{\rm free}$ is the free Dirac Hamiltonian on the Fock space of a variable number of $1 \leq n \leq N$ particles and the interaction Hamiltonian is given by $H_{\rm int} = H_{\rm int}^{\rm ann} + H_{\rm int}^{\rm cre}$, i.e.\
\be
	H_{\rm int} = \int dz \sum_{rst} \left( A_r^{st} a^\dagger_r(z) a_s(z) a_t(z) + (A_r^{st})^* a^\dagger_t(z) a^\dagger_s(z) a_r(z) \right).
\label{eq:hint}
\ee
We have already seen that the annihilation parts of the two models agree. It remains to study the creation part. To treat this part, we now show at the example $n=2$ that $H^{\rm cre}_{\rm int}$ gives rise to the IBC when suitably interpreted (for $n>2$ one proceeds analogously).  For $n=2$, \eqref{eq:trilala} becomes: 
\be
	 ( H^{\rm cre}_{\rm int} \varphi)^{(2)}_{s_1s_2}(t; z_1,z_2) =\frac{1}{\sqrt{2}} \sum_{r} \left( -(A^{s_1 s_2}_r)^* + (A^{s_2 s_1}_r)^* \right) \delta(z_1-z_2) \varphi^{(1)}_{r}(z_1).
\ee
Considering \eqref{eq:N}, \eqref{eq:ntilde}, one can see that the creation term vanishes for $s_1=s_2$, and that $ (A^{+-}_r)^* = e^{i \phi}  (A^{-+}_r)^*$. Anti-symmetry of $(H^{\rm cre}_{\rm int} \varphi)^{(2)}_{s_1s_2}(t; z_1,z_2)$ dictates $\phi=\pi$. We specialize to $s_1=-1, s_2=+1$; the reversed case leads to the same conclusions. Then:
\be 
	 ( H^{\rm cre}_{\rm int} \varphi)^{(2)}_{-+}(t; z_1,z_2) =-\sqrt{2} \sum_{r}  (A^{-+}_r)^* \delta(z_1-z_2) \varphi^{(1)}_{r}(t;z_1).
\ee
In order to give a proper interpretation to the $\delta$-function, we will integrate the corresponding Dirac equation 
\begin{equation}
 i \partial_t \varphi^{(2)} = H^{\mathrm{free}} \varphi^{(2)} + (H_{\rm int} \varphi)^{(2)}. \end{equation}
in a small neighborhood of the set where $z_1=z_2$. The component with $s_1=-1,s_2=+1$ reads
\begin{equation} \label{eq:diraccomponent}
i \partial_t \varphi_{-+}^{(2)}(t;z_1,z_2) =   \left( -i \partial_{z_1} + i \partial_{z_2} \right) \varphi^{(2)}_{-+}(t;z_1,z_2) + (H_{\rm int} \varphi)_{-+}^{(2)}(t;z_1,z_2).
\end{equation}
It is helpful to use relative coordinates $z = z_1 - z_2, Z = \tfrac{1}{2}(z_1+z_2)$ because we expect by \eqref{eq:2modelIBC} a jump discontinuity of $\varphi^{(2)}$ exactly at $z=0$. Note that $-i \partial_{z_1} + i \partial_{z_2}  = -2i \partial_z$. We integrate \eqref{eq:diraccomponent} over $dz$ from $- \varepsilon$ to $+ \varepsilon$, and let $\varepsilon$ go to zero. All terms vanish except the $\delta$-function and the derivatives w.r.t.\ $z$, which means 
\be
0 = \lim_{\varepsilon \to 0} \int_{ -\varepsilon}^{ \varepsilon} dz \left( -2i \partial_{z} \varphi^{(2)}_{-+}(t,Z+z/2,Z-z/2) + ( H^{\rm cre}_{\rm int} \varphi)^{(2)}_{-+}(t,Z+z/2,Z-z/2)\right).
\ee
For $s_1=-1,s_2=+1$ this becomes, omitting the common time variable $t$,
\be 
	 -i \lim_{\varepsilon \to 0} \left( \varphi^{(2)}_{-+}(Z+\varepsilon,Z-\varepsilon)- \varphi^{(2)}_{-+}(		
	 Z-\varepsilon,Z+\varepsilon) \right) = -\frac{1}{\sqrt{2}}  \sum_{r} (A^{-+}_r)^* \varphi^{(1)}_{r}(Z).
\ee
Using the anti-symmetry of $\varphi$, we arrive at 
\be  \label{eq:ibcfromdelta}
	 \lim_{\varepsilon \to 0} \left( \varphi^{(2)}_{-+}(Z-\varepsilon,Z+\varepsilon)+ \varphi^{(2)}_{+-}(Z-\varepsilon,Z+\varepsilon) \right) =  -\frac{1}{\sqrt{2} i} \left( w_1 \, \varphi^{(1)}_-(Z) + w_2 \,  \varphi^{(1)}_+(Z) \right).
	 \ee
We now compare this with the IBC \eqref{eq:2modelIBC} for $\phi=\pi$. Setting in addition $\theta = \pi$, we have:
\begin{equation}
\begin{split}
 \lim_{\varepsilon \to 0} \left( \varphi^{(2)}_{-+}(Z-\varepsilon,Z+\varepsilon)+ \varphi^{(2)}_{+-}(Z-\varepsilon,Z+\varepsilon) \right) &= \frac{1}{\sqrt{2}} B \varphi^{(1)}(Z) = -\frac{1}{2\sqrt{2} i} (1, -1) \widetilde{A}  \varphi^{(1)}(Z) 
 \\ & = -\frac{1}{\sqrt{2} i} \left( w_1 \, \varphi^{(1)}_-(Z) + w_2 \, \varphi^{(1)}_+(Z) \right),
\end{split}
\end{equation}
in agreement with \eqref{eq:ibcfromdelta}. 

The fact that we obtain the IBC \eqref{eq:2modelIBC} only for $\theta = \pi$ can be explained as follows. In \cite{LN:2015} it was shown that $\theta = \pi$ corresponds to the non-interacting case if no coupling between different sectors is present while $\theta \neq \pi$ leads to point interactions. For our model, this suggests that the case $\theta = \pi$ corresponds to interactions purely through particle exchange while $\theta \neq \pi$ includes additional point interactions. However, in the above argument, our starting point was $H_{\rm int}^{\rm ann}$, and this part of the Hamiltonian does not contain additional point interactions. In order to obtain the IBC for $\theta \neq \pi$, one would have to add point interactions to $H_{\rm int}$ manually. The IBC approach, on the other hand, incorporates the possibility of additional point interactions from the very start, and the case $\theta \neq \pi$ occurs more naturally.

Next, we discuss the Lorentz invariance of our model.

\section{Lorentz invariance} \label{sec:li}

In this section we discuss the behavior of our model under (proper) Lorentz transformations.
Lorentz invariance here concerns several aspects:
\begin{enumerate}
	\item Covariance of the wave function,
	\item Invariance of the domain,
	\item Probability conservation in all Lorentz frames,
	\item Invariance of the equations of motion,
	\item Invariance of the boundary conditions.
\end{enumerate}
Item 1.\ is clear because a multi-time wave function is a manifestly covariant object (see Eq. \eqref{eq:psitrafo}). In our case, $S[\Lambda]$ is given by the standard spinorial representation for the Dirac equation. 2.\ is also ensured as the set $\spacelike$ of spacelike configurations is invariant under Lorentz transformations. ($\spacelike_1$ is also invariant under proper Lorentz transformations but not under reflections.) Concerning 3., we have already established in Thm. \ref{thm:localcurrentcons} that our model leads to probability conservation on all Cauchy surfaces which include the equal-time surfaces of all frames. We shall discuss 4.\ and 5.\ now.

An element $\Lambda$ of the proper Lorentz group $\mathcal{L}^\uparrow_+$ in $d=1$ is a boost in the only existing spatial direction, characterized by a parameter $\beta \in \R$. Under $\Lambda$, the multi-time wave function transforms according to \eqref{eq:psitrafo}.
In 1+1 dimensions and with our choice of the basis in spin space, we have:
\begin{equation}
S[\Lambda]_s^{s'} ~=~ \delta_s^{s'} \left(\cosh(\beta/2) - s \sinh(\beta/2)\right).
\end{equation}
We shall check whether the transformed wave function $\psi'$ solves the primed versions of Eqs. \eqref{eq:multitimeeqwithsources} and \eqref{eq:mgeneralibcs}. Indeed, as a consequence of these equations one finds:
\begin{equation} \begin{split}
i(\partial_{t_k}-s_k \partial_{z_k}) {\psi'}^{(n)}_{s_1...s_n}(x_1,...,x_n) ~= &  \sum_{t,u=\pm 1} (-1)^k \left(\cosh(\beta/2) - s_k \sinh(\beta/2)\right) A^{tu}_{s_k} \times \\ & \times {\psi'}^{(n+1)}_{s_1...s_{k-1}tus_{k+1}...s_n}(x_1,...,x_k,x_k,x_{k+1},...,x_n). \end{split}
\end{equation}
Here we have used that $A^{tu}_{s_k}=0$ whenever $t=u$ and that $S[\Lambda]^+_+ S[\Lambda]^-_-$ cancels out because $\left(\cosh(\beta/2) -  \sinh(\beta/2)\right)\left(\cosh(\beta/2) +\sinh(\beta/2)\right)=1$. The equation would be Lorentz invariant if the matrix $A$ transformed like a spinor, with its upper indices transformed via $S^{-1}[\Lambda]$, i.e.
\begin{equation} \begin{split}
{A'}^{tu}_{s_k} ~&=~ \left(\cosh(\beta/2) + t \sinh(\beta/2)\right) \left(\cosh(\beta/2) + u \sinh(\beta/2)\right) \left(\cosh(\beta/2) - s_k \sinh(\beta/2)\right) A^{tu}_{s_k} \\ & =~ \left(\cosh(\beta/2) - s_k \sinh(\beta/2)\right) A^{tu}_{s_k}. \end{split}
\end{equation}
Since $A$ does not transform in this way, but is a fixed matrix, Lorentz invariance is broken in this regard. The situation is similar for the IBC, where \eqref{eq:mgeneralibcs} implies
\begin{align}
		&\left( {\psi'}^{(n+1)}_{s_1...s_{k-1} -+ s_{k+1} ... s_n}- e^{i\theta}{\psi'}^{(n+1)}_{s_1...s_{k-1} +- s_{k+1} ... s_n} \right) (x_1,...,x_k,x_k, x_{k+1},...,x_n)\nonumber\\
	 &=~ \sum_s  \left(\cosh(\beta/2) + s \sinh(\beta/2)\right)  B^s {\psi'}^{(n)}_{s_1 ... s_{k-1} s \, s_{k+1} ... s_{n}}(x_1,...,x_n)
	\end{align}
Thus, if $B$ transformed like a spinor, the model would be manifestly Lorentz invariant.
One can now clearly see that the only point where Lorentz invariance fails is the occurrence of the constant matrices $A$ and $B$. This is due to the simplification that we only consider fermions. The matrix $A$, for example, needs to be introduced to match the number of spin components of $A \psi^{(n+1)}$ to the one of $\psi^{(n)}$. For more realistic QFTs with appropriate types of bosons (e.g.\ photons) as exchange particles this situation would not occur, and consequently there would be no issue with Lorentz invariance. To formulate such a multi-time IBC model with bosons as exchange particles as well as to address the problems that come along with it (such as the question of a suitable position representation for photons) is left as a task for future work.

In this context, we note that there are Lorentz invariant models of self-interacting fermions in 1+1 dimensions, such as the Thirring model \cite{thirring}. We have chosen a different model here which corresponds to cubic terms in the field operators (instead of quartic terms such as $(\overline{\psi} \gamma_\mu \psi) (\overline{\psi} \gamma^\mu \psi)$ as in the Thirring model) because the cubic coupling is, even though not fully Lorentz invariant, closer to quantum electrodynamics and the view of photons mediating the interaction between electrons. That decision derives from the fact that the core of the motivation for studying IBCs is the hope that these might contribute to a rigorous formulation of realistic QFTs free of the UV problem. Nevertheless, it would be a point of interest in its own right (and an idea for future research) to investigate if a multi-time formulation of the Thirring model (or a related model) using IBCs is feasible as well.

\section{Proofs} \label{sec:proofs}

\subsection{Proof of theorem \ref{thm:localcurrentcons}} \label{sec:prooflocalcurrentcons}

	We start with statement 1. The proof is based on a technique developed in \cite{lienert:2015a,LN:2015}. Let $\Sigma_1, \Sigma_2 \subset \R^2$ be smooth Cauchy surfaces given by time functions $\tau_i(z) : \R \rightarrow \R$, i.e.:
\be
	\Sigma_i = \{ (t,z) \in \R^2 : t = \tau_i(z)\},~~i=1,2.
\ee
We shall show \eqref{eq:probcons} in the following form:
	\be
		\sum_{n=1}^N \int_{\Sigma_1^n\cap \mathscr{S}^{(n)}_1} \omega^{(n)} = \sum_{n=1}^N \int_{\Sigma_2^n\cap \mathscr{S}^{(n)}_1} \omega^{(n)}.
	\label{eq:probcons2}
	\ee
To show \eqref{eq:probcons2}, we consider each sector $n$ separately and construct a closed surface $S^{(n)}$ to which we can apply Stokes' theorem. As $j$ is compactly supported in the spatial directions, we can choose $R>0$ such that for all $n = 1,...,N$, $j^{\mu_1...\mu_n}(x_1,...,x_n) = 0$ if there is a $k \in \{1,...,n \}$ such that the variable $z_k$ in $x_k=(t_k,z_k)$ satisfies $|z_k|>R$ . We define
\be
	\Sigma_i^R = \{ (t,z) \in \Sigma_i : |z| < R\},~~i=1,2.
\ee
Then,
\be
	\int_{\Sigma_i^n\cap \mathscr{S}^{(n)}_1} \omega^{(n)} = \int_{(\Sigma_i^R)^n\cap \mathscr{S}^{(n)}_1} \omega^{(n)}.
\label{eq:sigmarint}
\ee
Now consider the configuration spacetime volume
\be
	V_R^{(n)} = \left\{ (t_1,z_1;...;t_n,z_n) \in \overline{\mathscr{S}}_1^{(n)} \left| \, \exists s \in [0,1] : \begin{array}{l}
	 \forall i: t_i = \tau_1(z_i) + s(\tau_2(z_i) - \tau_1(z_i)) \\ {\rm and}~ |z_i| \leq R
\end{array}	 \right. \right\}.
\ee
$V_R^{(n)}$ is a bounded and closed, hence compact ($n$+1)-dimensional submanifold of $\R^{2n}$.  Its boundary $\partial V_R^{(n)}$ has the form
\be
	\partial V_R^{(n)} =(\Sigma_1^R)^n \cup (\Sigma_2^R)^n \cup M^{(n)}_R
\ee
and $M^{(n)}_R \subset \overline{\mathscr{S}}_1^{(n)}$ has the two parts
\be
	M^{(n)}_R = M_{1,R}^{(n)} \cup M_{2,R}^{(n)}
\ee
with
\be
	M_{1,R}^{(n)} = \{ (t_1,z_1;...;t_n,z_n) \in \partial V_R^{(n)} |\, \exists \, i : |z_i| = R\},
\ee
hence $j^{\mu_1...\mu_n} = 0$ on $M_{1,R}^{(n)}$, and
\be
	M_{2,R}^{(n)} = \{ (t_1,z_1;...;t_n,z_n) \in \partial V_R^{(n)} \, |~ \exists\, i: (t_i,z_i) = (t_{i+1},z_{i+1})  \}.
\label{eq:m2}
\ee
In this situation, we can apply Stokes' theorem to obtain:
\be
	\int_{\partial V^{(n)}_R} \omega^{(n)} = \int_{V_R^{(n)}} d \omega^{(n)},
\ee
and, as $j^{\mu_1...\mu_n} = 0$ on $M_{1,R}^{(n)}$, we find (considering \eqref{eq:sigmarint} as well as orientation conventions):
\be
	\int_{\Sigma_1^n \cap  \mathscr{S}^{(n)}_1} \omega^{(n)} -\int_{\Sigma_2^n \cap  \mathscr{S}^{(n)}_1} \omega^{(n)} = \int_{V_R^{(n)}} d \omega^{(n)} -\int_{M_{2,R}^{(n)}} \omega^{(n)}.
\label{eq:probbalance}
\ee
Summation over $n=1,...,N$ yields:
\be
	\sum_{n=1}^N \int_{\Sigma_1^n \cap  \mathscr{S}^{(n)}_1} \omega^{(n)} - \sum_{n=1}^N \int_{\Sigma_2^n \cap  \mathscr{S}^{(n)}_1} \omega^{(n)} = \sum_{n=1}^N \left(\int_{V_R^{(n)}} d \omega^{(n)} -\int_{M_{2,R}^{(n)}} \omega^{(n)}\right).
\label{eq:summedprobbalance}
\ee
We now show that condition \eqref{eq:localcurrentcons} makes the right hand side vanish. To this end, note that $M_{2,R}^{(1)} = \emptyset$. Furthermore, we have $d \omega^{(N)} = 0$ by assumption. Thus, if we can show
\be
	\int_{V_R^{(n)}} d \omega^{(n)} = \int_{M_{2,R}^{(n+1)}} \omega^{(n+1)}
\label{eq:probbalanceproof}
\ee
we obtain a telescoping sum and therefore, indeed
\be
	\sum_{n=1}^N \left(\int_{V_R^{(n)}} d \omega^{(n)} -\int_{M_{2,R}^{(n)}} \omega^{(n)}\right) = 0.
\ee
We turn to the proof of \eqref{eq:probbalanceproof}. Considering \eqref{eq:m2}, we have:
\be
	M_{2,R}^{(n+1)} = \bigcup_{k=1}^{n} R_k^{(n+1)},
\label{eq:m2decomp}
\ee
where
\be
	R_k^{(n+1)} = \{ (t_1,z_1;...;t_{n+1},z_{n+1}) \in M_{2,R}^{(n+1)} : (t_k,z_k) = (t_{k+1},z_{k+1}) \}.
\ee
Then, noting that $\Phi_k: (x_1,...,x_k,x_k,x_{k+1},...,x_n) \mapsto (x_1,...,x_n)$ from \eqref{eq:phik} defines a bijective map between $R_k^{(n+1)}$ and $V_R^{(n+1)}$, we obtain:
\begin{align}
	\int_{M_{2,R}^{(n+1)}} \omega^{(n+1)} &\stackrel{\eqref{eq:m2decomp}}{=} \sum_{k=1}^n \int_{R_k^{(n+1)}} \omega^{(n+1)}\nonumber\\
	&= \sum_{k=1}^n  \int_{V_R^{(n)}} {\Phi_k}^* \, \omega^{(n+1)}.
\end{align}
This equals $\int_{V_R^{(n)}} d\omega^{(n)}$ if
\be
	\sum_{k=1}^n {\Phi_k}^* \, \omega^{(n+1)} = d\omega^{(n)}.
\ee
This, in turn, is ensured by condition \eqref{eq:localcurrentcons}. To summarize, if  \eqref{eq:localcurrentcons} holds, we obtain \eqref{eq:probbalanceproof} and hence \eqref{eq:probcons2} which is equivalent to \eqref{eq:probcons}.

\vspace{0.2cm}
We now turn to point 2.  To this end, we explicitly compute both sides of \eqref{eq:localcurrentcons}. Denoting omission by $\widehat{(\cdot)}$, we find:
\begin{align}
	{\Phi_k}^* \, \omega^{(n+1)} = & \sum_{\mu_1,...,\widehat{\mu}_k,\widehat{\mu}_{k+1},...,\mu_{n+1}} (-1)^{\mu_1+ \cdots +\widehat{\mu}_k + \widehat{\mu}_{k+1} + \cdots + \mu_{n+1}} \, j^{\mu_1 ... \mu_{n+1}} \nonumber\\
&~~~~~~~~~~d x_1^{1-\mu_1} \wedge \cdots \wedge ( - dx_k^0 \wedge dx_k^1 - dx_k^1 \wedge dx_k^0) \wedge \cdots \wedge dx_n^{1-\mu_{n+1}}\nonumber\\
= & \sum_{\mu_1,...,\widehat{\mu}_k,\widehat{\mu}_{k+1},...,\mu_{n+1}}  (-1)^{1+\mu_1+ \cdots +\widehat{\mu}_k + \widehat{\mu}_{k+1} + \cdots + \mu_{n+1}} \nonumber\\
&~~~~~~~\varepsilon_{\rho\, \sigma}\,  j^{\mu_1 ... \mu_{k-1} \, \rho \, \sigma \,  \mu_{k+1} ... \mu_{n+1}} \, d x_1^{1-\mu_1} \wedge \cdots \wedge dx_k^0 \wedge dx_k^1 \wedge \cdots \wedge dx_n^{1-\mu_{n+1}}.
\label{eq:restrictedcurrentform}
\end{align}
The left hand side of \eqref{eq:localcurrentcons} is given by:
\begin{align}
	d\omega^{(n)} ~=~ &\sum_{k=1}^n \sum_{\mu_1,...,\mu_n} (-1)^{\mu_1+\cdots \mu_k+ \cdots\mu_n} \, \partial_{k,\mu_k} j^{\mu_1 ...\mu_k... \mu_n} \nonumber\\
& (-1)^{(k-1) + \mu_k} dx_1^{1-\mu_1} \wedge \cdots \wedge dx_k^0 \wedge dx_k^1 \wedge \cdots \wedge dx_n^{1-\mu_n}.
\end{align}
Thus, ${\rm d} \omega^{(N)} = 0$ is equivalent to the first line of \eqref{eq:currentcondition}. Comparing $\sum_{k=1}^n {\Phi_k}^* \, \omega^{(n+1)}$ and $d\omega^{(n)}$, we obtain the following condition (relabelling indices $\mu_{k+1}...\mu_{n+1} \rightarrow \mu_k ... \mu_n$ in \eqref{eq:restrictedcurrentform}):
\be
	\varepsilon_{\rho \, \sigma}\,  j^{\mu_1 ... \mu_{k-1}\, \rho \, \sigma \, \mu_{k+1} ... \mu_n}(x_1,...,x_k,x_k,x_{k+1},...,x_n) = (-1)^k \partial_{k,\mu_k} j^{\mu_1 ...\mu_k... \mu_n}(x_1,...,x_n).
\ee
This is identical to the second line of \eqref{eq:currentcondition}. \qed

\subsection{Proof of theorem \ref{thm:2ibc}} \label{sec:proof2ibc}

	We start from the balance condition \eqref{eq:2currentbalance}. As we aim at a translation invariant model \eqref{eq:2modellower}-\eqref{eq:2modelIBC}, it is clear that the phase $\theta$ and the matrices $A, B$ must be constant.

For ease of notation, we omit the arguments of the wave function. In order to simplify \eqref{eq:2currentbalance}, we eliminate the component $\psiii_{-+}$ from the equation using the IBC \eqref{eq:2modelIBC}. Introducing
\be
	\psitilde = \left( \begin{array}{c} \psii_- \\ \psii_+\\ \psiii_{+-}\end{array} \right),
\ee
\eqref{eq:2currentbalance} can be rewritten as:
\begin{align}
	&\frac{1}{2i} \, \psitilde^\dagger \left[ \left( \begin{array}{cc} \widetilde{A}^\dagger  \left( \begin{array}{c} B\\ 0~0\end{array} \right) & \widetilde{A}^\dagger \left( \begin{array}{c} e^{i\theta}\\ 1 \end{array}\right) \\ ~&~\\0~~~~~~~~0 & 0 \end{array}\right)  - \left( \begin{array}{cc} \left( B^\dagger \begin{array}{c}  0\\0\end{array} \right)  \widetilde{A} & \begin{array}{c} 0\\ 0\end{array} \\ ~&~\\ (e^{-i\theta}, 1) \, \widetilde{A} & 0 \end{array}\right)\right] \psitilde \nonumber\\
	&= ~ \psitilde^\dagger \left( \begin{array}{cc} -B^\dagger B & -B^\dagger e^{i\theta}\\ -e^{-i\theta} B& 0 \end{array}\right) \psitilde.
\end{align}
As this equation has to hold for all $\psitilde$, we obtain a condition for the matrices on l.h.s. and r.h.s. Evaluating this condition in detail, we obtain just two independent conditions: Firstly,
	\be
		 \frac{1}{2i}(e^{-i\theta},1) \, \widetilde{A} = e^{-i\theta} B ~~~
		 \Leftrightarrow ~~~ B = \frac{1}{2i} (1,e^{i\theta}) \, \widetilde{A}.
		 \label{eq:mofn}
	\ee
	This yields \eqref{eq:M}. Secondly,
	\be
		 \frac{1}{2i} \left[ \widetilde{A}^\dagger \left( \begin{array}{c} B\\ 0~0\end{array} \right) - \left( B^\dagger \begin{array}{c}  0\\0\end{array} \right) \widetilde{A} \, \right] = -B^\dagger B.
	\ee
	Plugging \eqref{eq:mofn} into this equation, we obtain the condition
	\begin{align}
		\frac{1}{4} \, \widetilde{A}^\dagger \left( \begin{array}{cc} 2 & e^{i\theta}\\ e^{-i\theta} & 0 \end{array} \right) \widetilde{A} ~&=~ \frac{1}{4} \, \widetilde{A}^\dagger \left( \begin{array}{cc} 1 & e^{i\theta}\\ e^{-i\theta} & 1 \end{array} \right) \widetilde{A} \nonumber\\
		\Leftrightarrow~~~~~~~~~  \widetilde{A}^\dagger \left( \begin{array}{cc} 1 & 0\\ 0 & -1 \end{array} \right) \widetilde{A} ~&=~ 0.
		\label{eq:condn}
	\end{align}
Let $\widetilde{A} = \left( \begin{array}{cc} a&b\\ c&d\end{array}\right)$. \eqref{eq:condn} then yields the three conditions (i)$ |a| = |c|$, (ii) $|b| = |d|$ and (iii) $a^* b = c^* d$. These force $\widetilde{A}$ to be a rank-1 matrix of the form \eqref{eq:ntilde}. 
\qed

\subsection{Proof of theorem \ref{thm:2main}} \label{sec:proof2main}

We now prove the existence and uniqueness of solutions for our model with two sectors of Fock space. Later we shall also do this for $N$ sectors; however, the case $N=2$ is crucial to develop the technique of the proofs and makes the proof for a general $N>2$ much more transparent.

The proof is divided into two steps. First we show that:
\begin{itemize}
\item Given the wave function in sector 1, we obtain a unique solution for sector 2.
\item Given the wave function in sector 2, we obtain a unique solution for sector 1.
\end{itemize}
Second, we use a fixed point argument to find a combined solution of both sectors. The first step is carried out in the following lemmas.

\begin{lemma} \label{thm:firstlemma}
Let \begin{equation} \label{eq:2sectordefB} \begin{split}
\Banach_1& :=  C_b^1 \left( [-T,T] \times \R, \C^2 \right),
\\ \Banach_2& := C_b^1 \left( \{ (t_1,z_1,t_2,z_2) \in \mathscr{S}_1^{(2)} | t_1,t_2 \in [-T,T],  z_1 < z_2  \}, \C^4 \right).
\end{split}
\end{equation}
Then, given a function $f^{(2)} \in \Banach_2$ and initial values $\psi^{(1)}(0,\cdot) \in C_b^1(\R,\C^2)$, there exists a unique solution $\psi^{(1)} \in \Banach_1$ of
\begin{equation} \label{eq:lowersectorwithf}
i \partial_t \psii(t,z) = H^{\rm Dirac}_1 \psii(t,z) - A f^{(2)} (t,z,t,z)
\end{equation}
with these initial values.
\end{lemma}

\begin{proof}
Rewriting the system \eqref{eq:lowersectorwithf} gives 
\begin{equation}
	\left( \begin{array}{c} (\partial_t + \partial_z) \psi^{(1)}_-(t,z) \\
	(\partial_t - \partial_z) \psi^{(1)}_+(t,z)
	\end{array} \right) = i A f^{(2)}(t,z,t,z) =: \left( \begin{array}{c} \widetilde{f}_-(t,z) \\ \widetilde{f}_+(t,z)  
	\end{array} \right) ,
	\label{eq:deffpm}
\end{equation}
with $\widetilde{f}_\pm \in C_b^1(\R^2,\C)$. Now \eqref{eq:deffpm} can be directly integrated along characteristics. The solution is given by:
\begin{equation}
\psi^{(1)}(t,z) = \left( \begin{array}{c}\psi^{(1)}_-(0,z-t) \\
	\psi^{(1)}_+(0,z+t)
	\end{array} \right) + \int_0^t ds \left( \begin{array}{c}\widetilde{f}_-(s,z-t+s) \\
	\widetilde{f}_+(s,z+t-s)
	\end{array} \right). \label{eq:solutionformula1}
\end{equation}
We have $\psi^{(1)}\in \Banach_1$ because the initial values and $\widetilde{f}$ are $C_b^1$-functions.\qed
\end{proof}

\begin{lemma} \label{thm:secondlemma}
Given a function $\psi^{(1)} \in \Banach_1$ and initial values
\begin{equation}
\psi^{(2)}(0,\cdot,0,\cdot) = \psi^{(2)}_0 \in C_b^1 \left( \{ (z_1,z_2) \in \R^2 | z_1 < z_2  \}, \C^4 \right) \end{equation}
which satisfy \eqref{eq:IBCforinitialvalues} as well as the condition
\begin{equation} \label{eq:2modelconditionC1}
\left. B \,\partial_t \psii(t,z)\right|_{t=0} = \left( \partial_{z_2} - \partial_{z_1} \right) \left. \left( \psiii_{0,-+}(z_1,z_2)+e^{i\theta}\psiii_{0,+-}(z_1,z_2)  \right) \right|_{z_1=z_2=z}
\end{equation}
there exists a unique solution $\psi^{(2)} \in \Banach_2$ of \eqref{eq:2modelupper} with boundary condition \eqref{eq:2modelIBC}.
\end{lemma}

\begin{proof}
The lemma is a special case of theorem 3.3. in \cite{lienert:2015a}, where the solution (eq. (32) in that paper) was given explicitly for a general class of initial boundary value problems. Adapted to our notation and with the characteristic variables $u_k := z_k - t_k$ and $v_k := z_k + t_k$ for $k=1,2$, the solution of \eqref{eq:2modelupper} on $\mathscr{S}^{(2)}_1$ with boundary condition \eqref{eq:2modelIBC} reads as follows:
\begin{equation} \label{eq:2sectorexplicitsol}
	\begin{split}
	\psi^{(2)}_{--}(t_1,z_1,t_2,z_2) & =  \psi^{(2)}_{0,--} (u_1,u_2) ,
	\\ \psi^{(2)}_{-+}(t_1,z_1,t_2,z_2) & = \left\lbrace \begin{array}{cl}
\psi^{(2)}_{0,-+}(u_1,v_2) & \text{for} \ u_1 < v_2,
\\ e^{i\theta} \psi^{(2)}_{0,+-}(v_2,u_1) + B \psi^{(1)}(\tfrac{v_2-u_1}{2}, \tfrac{v_2+u_1}{2}) & 	\text{for} \ u_1 \geq v_2,
\end{array}	 \right. 
	\\ \psi^{(2)}_{+-}(t_1,z_1,t_2,z_2) & = \left\lbrace \begin{array}{cl}
\psi^{(2)}_{0,+-}(v_1,u_2) & \text{for} \ v_1 < u_2,
\\ e^{-i\theta} \left( \psi^{(2)}_{0,-+}(u_2,v_1) - B \psi^{(1)}(\tfrac{v_1-u_2}{2}, \tfrac{v_1+u_2}{2}) \right) & 	\text{for} \ v_1 \geq u_2,
\end{array}	 \right. 
	\\  \psi^{(2)}_{++}(t_1,z_1,t_2,z_2) & =  \psi^{(2)}_{0,++} (v_1,v_2).
	\end{split}
\end{equation}
The such defined $\psiii$ inherits the $C_b^1$-property from the initial and boundary values wherever $u_k \neq v_j$ for $j \neq k$. At $u_1=v_2$ resp.\ $v_1=u_2$, continuity of $\psiii_{-+}$  resp.\ $\psiii_{+-}$ amounts to condition \eqref{eq:IBCforinitialvalues}. In order to check differentiability at those points, we compare the limits of the respective partial derivatives in the case differentiation in \eqref{eq:2sectorexplicitsol}. We start with comparing $\partial_{z_1} \psiii_{-+}$ for $u_1 \nearrow v_2$ and $u_1 \searrow v_2$ at $v_2=z$. Let $D_k$ denote the derivative w.r.t.\ the $k$-th argument. The condition for the two limits to coincide then is:
\begin{equation}
	D_1 \psiii_{0,-+}(z,z) ~=~ e^{i \theta}D_2 \psiii_{0,+-}(z,z) + \tfrac{1}{2} B (D_2 \psii(0,z) - D_1 \psii(0,z) ).
\end{equation}
Inserting the $z$-derivative of \eqref{eq:IBCforinitialvalues},
\begin{equation}
(D_1 + D_2) (\psiii_{0,+-}(z,z) - e^{i \theta} \psiii_{0,+-}(z,z)) = B D_2 \psii(0,z),
\end{equation}
this becomes \eqref{eq:2modelconditionC1}. Similar computations show that all other partial derivatives exist and are continuous as a result of the same conditions.
 \qed
\end{proof}

With the lemmas at hand, we now construct a fixed point map for our model. For given initial values \eqref{eq:2modelinitial} satisfying the compatibility conditions \eqref{eq:IBCforinitialvalues} and \eqref{eq:IBCforinitialvalues2}, we let
\be
	\mathcal{D}= \left\lbrace (\psi^{(1)},\psi^{(2)}) \in \Banach_1 \oplus \Banach_2  \, \big| \, \psi^{(1)}(0,z)=\psii_0(z), \psiii(0,z_1,0,z_2)=\psiii_0(z_1,z_2) \right\rbrace.
\ee
Clearly, $\mathcal{D}$ is closed in $\Banach_1 \oplus \Banach_2$. 

\begin{definition} Let $F: \mathcal{D} \to \mathcal{D},~ (f^{(1)}, f^{(2)}) \mapsto (\psi^{(1)}, \psi^{(2)})$ be defined by the following procedure.
\begin{itemize}
\item Take $\psi^{(1)}$ to be the unique solution of \eqref{eq:lowersectorwithf} according to lemma \ref{thm:firstlemma}.
\item Using the such constructed $\psi^{(1)}$ in the interior-boundary condition, $\psi^{(2)}$ is defined to be the unique solution of \eqref{eq:2modelupper} and \eqref{eq:2modelIBC} according to lemma \ref{thm:secondlemma}.
\end{itemize}
\end{definition}

\paragraph{Well-definedness of $F$:} We need to check that $F$ actually maps into $\mathcal{D}$. This is true if the following points hold.
\begin{itemize}
\item \textit{$C^1_b$-property.} Lemma \ref{thm:secondlemma} gives a $C^1_b$-solution under the conditions \eqref{eq:IBCforinitialvalues} and \eqref{eq:2modelconditionC1}. We require \eqref{eq:IBCforinitialvalues} and \eqref{eq:IBCforinitialvalues2}. Since $\psii$ solves \eqref{eq:lowersectorwithf}, 
\begin{equation}  \label{eq:2sectorC1ofF}
B \partial_t \psii(t,z)|_{t=0} = -iB \left( H^{\rm Dirac} \psii(0,z) - Af^{(2)}(0,z,0,z) \right) 
\end{equation}
with $f^{(2)}(0,z,0,z) = \psiii_0(z,z)$ and inserting \eqref{eq:IBCforinitialvalues2} implies \eqref{eq:2modelconditionC1}.
\item \textit{Initial values.} These are preserved under $F$ as $(\psi^{(1)}, \psi^{(2)})=F(f^{(1)},f^{(2)})$ is constructed with respect to the same initial values.
\end{itemize}

By construction of $F$, we immediately obtain the following result.
\begin{lemma} \label{thm:lemmaF} Let $(\psi^{(1)}, \psi^{(2)}) \in \Banach_1 \oplus \Banach_2$. Then the following statements are equivalent: 
\begin{enumerate}
\item[i)] $(\psi^{(1)}, \psi^{(2)})$ is a $C^1_b$-solution of the initial boundary value problem \eqref{eq:2modelupper},\eqref{eq:2modellower}, \eqref{eq:2modelIBC} with initial values given as in \eqref{eq:2modelinitial}.
\item[ii)] $(\psi^{(1)}, \psi^{(2)})$ lies in $\mathcal{D}$ and is a fixed point of $F$.
\end{enumerate} 
\end{lemma}

The main work now lies in proving the following.

\begin{lemma} \label{thm:lemmafixedpoint}
For every $T >0$, $F$ possesses a unique fixed point in $\mathcal{D}$.
\end{lemma}

\begin{proof}
Let $\gamma \geq 0$. We equip $\Banach_1$ and $\Banach_2$ with the weighted norms
\begin{equation} \label{eq:2gammanorms}
\begin{split}
\left\| f^{(1)} \right\|_{\Banach_1, \gamma} & :=  \sup_{t \in [-T,T], z \in \R} \left( \Big( |f^{(1)}(t,z)| + \max_{y \in \{t,z\}} |\partial_y f^{(1)}(t,z)| \Big) e^{- \gamma |t|} \right),
\\ \left\| f^{(2)} \right\|_{\Banach_2, \gamma} & :=  \sup_{t_1,t_2 \in [-T,T], (t_1,z_1,t_2,z_2) \in \mathscr{S}^{(2)}_1} \bigg( \Big( |f^{(2)}(t_1,z_1,t_2,z_2)| \\ & \hspace{0.9cm}+ \max_{y \in \{t_1,z_1,t_2,z_2 \}} |\partial_y f^{(2)}(t_1,z_1,t_2,z_2)|\Big)  e^{-\tfrac{\gamma}{2} (|t_1| + |t_2|)} \bigg),
\end{split}
\end{equation}
where $| \cdot |$ denotes the maximum norm of $\C^2$ and $\C^4$, respectively.
For $\gamma = 0$, the norms \eqref{eq:2gammanorms} reduce to the canonical norms on $C_b^1$-functions. In that case, one obtains complete spaces. As a consequence of the inequalities
\begin{equation} 
\left\| f^{(k)} \right\|_{\Banach_k, 0} e^{- \gamma T} ~ \leq ~ \left\| f^{(k)} \right\|_{\Banach_k, \gamma} ~\leq~ \left\| f^{(k)} \right\|_{\Banach_k, 0}, \quad k=1,2,
\end{equation}
the norms \eqref{eq:2gammanorms} are equivalent for all $\gamma \geq 0$. This implies that $\Banach_1 \oplus \Banach_2$ equipped with the norm
\begin{equation}
\left\| (f^{(1)}, f^{(2)}) \right\|_{\Banach_1 \oplus \Banach_2,\gamma} :=~ \left\| f^{(1)} \right\|_{\Banach_1,\gamma} + ~ \left\| f^{(2)} \right\|_{\Banach_2,\gamma}
\end{equation}
is a Banach space. Recall that $\mathcal{D}$ is a closed subset of $\Banach_1 \oplus \Banach_2$.
Our goal is to use Banach's fixed point theorem, so it remains to show that $F: \mathcal{D} \to \mathcal{D}$ is a contraction.

Let $f,g \in \mathcal{D}$ and $F \left( f^{(1)}, f^{(2)} \right) =:  (\psi^{(1)}, \psi^{(2)})$ and $F \left( g^{(1)}, g^{(2)} \right) =:  (\phi^{(1)}, \phi^{(2)})$. Moreover, we define $\widetilde{g}$ like $\widetilde{f}$ in \eqref{eq:deffpm} but with $f$ replaced by $g$. Using the solution formula \eqref{eq:solutionformula1}, we obtain:
\begin{equation}
\left( \begin{array}{c}(\psi^{(1)}_- -  \phi^{(1)}_-)(t,z) \\
	(\psi^{(1)}_+ -  \phi^{(1)}_+)(t,z)
	\end{array} \right) =  \int_0^t ds \left( \begin{array}{c}(\widetilde{f}_- - \widetilde{g}_-)(s,z-t+s) \\
	(\widetilde{f}_+ - \widetilde{g}_+)(s,z+t-s)
	\end{array} \right).
\end{equation}
This implies:
\begin{equation} \begin{split}
\left|\psi^{(1)}_\pm - \phi^{(1)}_\pm\right|(t,z) ~&\leq~ \mathrm{sgn}(t) \int_0^t ds ~|(\widetilde{f}_\pm - \widetilde{g}_\pm)(s,z\pm t - \pm s) | \,  e^{-\gamma |s|} e^{\gamma |s|} 
\\ & \leq \sup_{s \in [-|t|,|t|]} \left( |\widetilde{f}_\pm - \widetilde{g}_\pm|(s, z \pm t - \pm s) \, e^{- \gamma |s|}  \right) \mathrm{sgn}(t) \int_0^t ds ~ e^{\gamma |s|}
\\ & \leq \sup_{s \in [-|t|,|t|], y \in \R}\left( |\widetilde{f}_\pm - \widetilde{g}_\pm|(s,y) \, e^{- \gamma |s|} \right) \frac{1}{\gamma} e^{\gamma |t|}.
 \end{split}
\end{equation}
Hence, recalling \eqref{eq:deffpm}, we find:
\begin{equation}  \label{eq:firstbound}
 \begin{split}
\left|\psi_\pm^{(1)} - \phi_\pm^{(1)}\right|(t,z) \, e^{- \gamma |t|}&  ~\leq~ \frac{1}{\gamma} \sup_{s \in [-|t|,|t|], y \in \R} |\widetilde{f}_\pm - \widetilde{g}_\pm|(s,y)\, e^{- \gamma |s|}
\\ & ~\leq~  \frac{1}{\gamma} \sup_{s \in [-|t|,|t|], y \in \R} \|A \|_{\infty} \, |f^{(2)} - g^{(2)}|(s,y,s,y) \, e^{-\gamma |s|}. \end{split}
\end{equation}
For the $z$-derivative, we obtain an analogous formula:
\begin{align} \label{eq:zderivative} 
\left|\partial_z (\psi^{(1)}_\pm - \phi^{(1)}_\pm)(t,z) \right|\, e^{-\gamma |t|} ~ &\leq~ \frac{1}{\gamma} \sup_{s \in [-|t|,|t|], y \in \R} \left|\partial_y(\widetilde{f}_\pm - \widetilde{g}_\pm)(s,y) \right| \, e^{-\gamma |s|} \nonumber\\
&\leq~ \frac{1}{\gamma} \sup_{s \in [-|t|,|t|], y \in \R} \|A \|_{\infty} \, \left|\partial_y(f^{(2)} - g^{(2)})(s,y,s,y) \right| \, e^{-\gamma |s|}.
\end{align}
In the estimate for $\partial_t (\psi^{(1)}_\pm - \phi^{(1)}_\pm)(t,z)$ we obtain a similar expression as \eqref{eq:zderivative} (with $\partial_y$ replaced by $\partial_s$) plus $(\widetilde{f}_\pm - \widetilde{g}_\pm)(t,z)$. The latter appears due to the time-dependent upper bound of the integral $\int_0^t ds$. We can bound it as follows.
\begin{align} \label{eq:2boundderivatives} 
\left|(\widetilde{f}_\pm - \widetilde{g}_\pm)(t,z)\right| ~& =~ |(\widetilde{f}_\pm - \widetilde{g}_\pm)(t,z) - (\widetilde{f}_\pm - \widetilde{g}_\pm)(0,z)| ~= ~\bigg| \int_0^t \partial_s (\widetilde{f}_\pm - \widetilde{g}_\pm)(s,z) ds \bigg|
\nonumber\\ & \leq~ \mathrm{sgn}(t) \int_0^t \big| \partial_s (\widetilde{f}_\pm - \widetilde{g}_\pm)(s,z) \big| e^{-\gamma |s|} e^{\gamma |s|} ds 
\nonumber\\ & \leq~  \sup_{s \in [-|t|,|t|], y \in \R} \left( |\partial_s( \widetilde{f}_\pm - \widetilde{g}_\pm)(s,z)| e^{-\gamma |s|} \right) \frac{1}{\gamma} e^{\gamma |t|}.
\end{align}
Here we have used that $\widetilde{f}_\pm(0,z) = \widetilde{g}_\pm(0,z)$ by definition of $\mathcal{D}$ and that $\widetilde{f}$ and $\widetilde{g}$ are $C^1_b$-functions.
\eqref{eq:2boundderivatives} implies:
\be \label{eq:additionalterm}
	\left|(\widetilde{f}_\pm - \widetilde{g}_\pm)(t,z)\right| e^{-\gamma |t|} ~\leq~  \frac{1}{\gamma}  \,  \|A \|_{\infty} \sup_{s \in [-|t|,|t|], y \in \R} \left| \partial_s(f^{(2)} - g^{(2)})(s,z,s,z)\right| e^{-\gamma |s|}.
\ee
 Gathering the previous estimates \eqref{eq:firstbound}, \eqref{eq:zderivative}, \eqref{eq:additionalterm} and considering \eqref{eq:2gammanorms}, we obtain the bound
\begin{equation} \label{eq:boundfirstnorm}
	\left\| \psi^{(1)} -  \phi^{(1)} \right\|_{\Banach_1,\gamma} ~\leq~   \frac{2}{\gamma} \, \| A \|_{\infty} \left\| f^{(2)} - g^{(2)} \right\|_{\Banach_2,\gamma}.
\end{equation}
To bound the norm for the second sector, recall the solution formula \eqref{eq:2sectorexplicitsol}. Since $\psi^{(2)}$ and $\phi^{(2)}$ have the same initial data $\psi^{(2)}_0$, their difference is given by (recall $u_k = z_k-t_k$, $v_k = z_k+t_k$):
\begin{align} \label{eq:difference}
(\psi^{(2)}-\phi^{(2)})(t_1,z_1,t_2,z_2) = \left( \begin{array}{c}
0 \\ \hspace*{-0.2cm} \big[ (B \psi^{(1)})_- - (B \phi^{(1)})_-\big] (\tfrac{v_2-u_1}{2},\tfrac{v_2+u_1}{2}) \, \id_{\{u_1 \geq v_2\} } (t_1,z_1,t_2,z_2)
\\ \hspace*{-0.2cm} \big[ (B \psi^{(1)})_+ - (B \phi^{(1)})_+\big](\tfrac{v_1-u_2}{2},\tfrac{v_1+u_2}{2}) \, \id_{\{v_1 \geq u_2\} } (t_1,z_1,t_2,z_2)  \\ 0
\end{array} \hspace*{-0.2cm} \right)
\end{align}
where we write $B \alpha = \left( \begin{array}{c}
(B \alpha)_- \\ (B \alpha)_+
\end{array} \right)$ for two-component vectors $\alpha$. For positive times $t_1,t_2>0$, only the third line of \eqref{eq:difference} is nonzero, and then the weight factors in the $\gamma$-norms satisfy
\begin{equation}
 e^{- \tfrac{\gamma}{2}(|t_1|+|t_2|)} = e^{- \tfrac{\gamma}{2}(t_1+t_2)}  \leq e^{- \tfrac{\gamma}{2}(t_1+t_2)} \underbrace{e^{- \tfrac{\gamma}{2} (z_1-z_2)}}_{\geq 1 \ \text{as}\  z_1 \leq z_2} = e^{-\gamma \tfrac{v_1-u_2}{2}}= e^{-\gamma \tfrac{|v_1-u_2|}{2}}.
\end{equation}
If one time is positive and the other is negative, $(\psi^{(2)}-\phi^{(2)})(t_1,z_1,t_2,z_2) =0$. For $t_1, t_2 < 0$, only the second line of \eqref{eq:difference} is nonzero, and the weight factors in the $\gamma$-norms satisfy
\begin{equation}
 e^{- \tfrac{\gamma}{2}(|t_1|+|t_2|)} = e^{ \tfrac{\gamma}{2}(t_1+t_2)}  \leq e^{\tfrac{\gamma}{2}(t_1+t_2)} \underbrace{e^{ \tfrac{\gamma}{2} ( z_2-z_1)}}_{\geq 1 \ \text{as}\  z_1 \leq z_2} \leq e^{\gamma \tfrac{v_2-u_1}{2}} =e^{-\gamma \tfrac{|v_2-u_1|}{2}}.
\end{equation}
Now $(v_1-u_2)/2$ and $(v_2-u_1)/2$ appear as the time arguments of $[ (B \psi^{(1)})_\mp - (B \phi^{(1)})_\mp]$, respectively, in \eqref{eq:difference}. It is therefore clear that
\begin{equation}
\left\| \psi^{(2)} -  \phi^{(2)} \right\|_{\Banach_2,\gamma} \leq~ \left\| B\psi^{(1)} -  B\phi^{(1)} \right\|_{\Banach_1,\gamma} \leq~ \left\| B \right\|_\infty \left\|  \psi^{(1)} -  \phi^{(1)} \right\|_{\Banach_1,\gamma}.
\end{equation}
Together with \eqref{eq:boundfirstnorm}, this implies that
\begin{align}
\left\| \big(\psi^{(1)}, \psi^{(2)}\big) - \big(\phi^{(1)}, \phi^{(2)}\big) \right\|_{\Banach_1 \oplus \Banach_2,\gamma} & \leq~ \big(\| B \|_\infty +1 \big) \left\|  \psi^{(1)} -  \phi^{(1)} \right\|_{\Banach_1,\gamma} \nonumber \\
& \leq~ \frac{2}{\gamma} \, \big(\| B \|_\infty +1 \big) \, \| A \|_\infty \left\| f^{(2)} - g^{(2)} \right\|_{\Banach_2,\gamma} \nonumber
\\ & \leq~ C \left\| \big(f^{(1)}, f^{(2)} \big) - \big(g^{(1)},g^{(2)} \big) \right\|_{\Banach_1 \oplus \Banach_2,\gamma} ,
\end{align}
with the constant $C =  \frac{2}{\gamma} (\| B \|_\infty +1) \| A \|_\infty$.

Choosing, for example, $\gamma = 10 \, (\| B \|_\infty +1) \| A \|_\infty$ we have $C= \tfrac{1}{5}<1$. Thus, $F$ is a contraction and Banach's fixed point theorem yields the claim. \qed
\end{proof}

Together with lemma \ref{thm:lemmaF} this proves theorem \ref{thm:2main}, establishing that the multi-time IBC system has a unique global $C_b^1$-solution for all times.

\subsection{Proof of theorem \ref{thm:mmain}} \label{sec:proofmmain}

As a preparation, we prove the following statement which expresses a certain harmony of the interaction terms in the multi-time equations for different sectors of Fock space.

\begin{lemma}[Consistency conditions.] \label{thm:nsectorsconsistency} The system of multi-time equations \eqref{eq:multitimeeqwithsources}, \eqref{eq:sourceterms} and \eqref{eq:mthsourceterm} satisfies the consistency conditions
	\be
		i(\partial_{t_k} - s_k \partial_{z_k}) f_{l,s_1...s_n}^{(n)} = i(\partial_{t_l} - s_l \partial_{z_l}) f_{k,s_1...s_n}^{(n)},
		\label{eq:consistencymsectors}
	\ee
for all $k,l =1,...,n$, $s_1,...,s_n = \pm 1$,
	\begin{enumerate}
	\item[(i)] for $n=N$ (trivially).	
	\item[(ii)] for $n=1,2,...,N-1$ if $\psi^{(n+1)}$ satisfies the multi-time equations \eqref{eq:multitimeeqwithsources}.

\end{enumerate}
\end{lemma}

\begin{proof}
	In the cases $n=1$ and $n=N$ there is nothing to show (note \eqref{eq:mthsourceterm} for $n=N$).
	 For $n=N-1$, we have, for all $k = 1,...,N-1$:
	\begin{align}
		&i(\partial_{t_k} - s_k \partial_{z_k}) f_{l,s_1...s_{N-1}}^{(N-1)}(x_1,...,x_{N-1}) \nonumber\\
		&\stackrel{\eqref{eq:sourceterms}}{=} \sum_{t ,u = \pm 1} (-1)^l A^{tu}_{s_l} i(\partial_{t_k} - s_k \partial_{z_k}) \psi^{(N)}_{s_1 ... s_{l-1} t \, u\, s_{l+1} ... s_{N-1}} (x_1,...,x_{l-1}, x_l, x_l, x_{l+1},...,x_{N-1}) \nonumber\\
		&= 0
	\end{align}
	because of \eqref{eq:mthsourceterm}. In particular, \eqref{eq:consistencymsectors} follows.
		
	 For $2 \leq n \leq N-2$, w.l.o.g. let $k < l$ and consider: 
	\begin{align}
		&i(\partial_{t_k} - s_k \partial_{z_k}) f_{l,s_1...s_{n}}^{(n)}(x_1,...,x_n) \nonumber\\
		&\stackrel{\eqref{eq:sourceterms}}{=} \sum_{t ,u = \pm 1} (-1)^l A^{tu}_{s_l} i(\partial_{t_k} - s_k \partial_{z_k}) \psi^{(n+1)}_{s_1 ... s_{l-1} t \, u\, s_{l+1} ... s_n} (x_1,...,x_{l-1}, x_l, x_l, x_{l+1},...,x_n) \nonumber\\
		&\stackrel{\eqref{eq:multitimeeqwithsources},\eqref{eq:sourceterms}}{=} \sum_{t ,u,v,w = \pm 1} (-1)^{k+l} A^{tu}_{s_l}A^{vw}_{s_k} \psi^{(n+2)}_{s_1 ... s_{k-1} v \, w\, s_{k+1} s_{l-1} t \, u\, s_{l+1} ... s_n} (x_1,...,x_{k-1}, x_k, x_k, x_{k+1},...,\nonumber\\
		&~~~~~~~~~~~~~~~~~~~~~~~~~~~~~~~~~~~~~~~~~~~~~~~~~
		~~~~~~~~~~~~~~~~~~~~~~~x_{l-1}, x_l, x_l, x_{l+1},...,x_n).
		\label{eq:cccalc1}
	\end{align}
	A similar calculation yields:
	\begin{align}
		&i(\partial_{t_l} - s_l \partial_{z_l}) f_{k,s_1...s_{n}}^{(n)}(x_1,...,x_n) \nonumber\\
		&= \sum_{t ,u,v,w = \pm 1} (-1)^{k+l} A^{tu}_{s_k} A^{vw}_{s_l} \psi^{(n+2)}_{s_1 ... s_{k-1} t \, u\, s_{k+1} s_{l-1} v \, w\, s_{l+1} ... s_n} (x_1,...,x_{k-1}, x_k, x_k, x_{k+1},...,\nonumber\\
		&~~~~~~~~~~~~~~~~~~~~~~~~~~~~~~~~~~~~~~~~~~~~~~~~~
		~~~~~~~~~~~~~~~~~~~~~~~x_{l-1}, x_l, x_l, x_{l+1},...,x_n).
		\label{eq:cccalc2}
	\end{align}
	Relabeling $t,u \leftrightarrow v,w$ shows that \eqref{eq:cccalc1}, \eqref{eq:cccalc2} agree; hence we obtain \eqref{eq:consistencymsectors}. \qed
\end{proof}

The idea now is to prove Thm. \ref{thm:mmain} using a fixed point argument. In fact, it is possible to explicitly write down the solution of the model for a particular sector provided given the wave function of the neighboring sectors. First we explain how to do this heuristically. Then we define the fixed point map and show that it is, indeed, a contraction in a sequence of lemmas.

\paragraph{Heuristics.} We now explain at the example $N=3$ how to obtain a solution of the multi-time equations for a particular sector, given the solution on the neighboring sectors. In our previous work \cite{LN:2015}, we constructed the solution for $A = 0$ such that probability is conserved for each sector, separately. This was done following the so-called \textit{multi-time characteristics} back to the initial value surface at time zero. The multi-time characteristic associated with a certain component $\psi^{(3)}_{s_1 s_2 s_3}$ and a particular point $(t_1,z_1;t_2,z_2;t_3,z_3) \in \spacelike^{(3)}_1$ is defined as the set that contains $(t_1,z_1;t_2,z_2;t_3,z_3)$ and along which that component would be constant by the homogeneous part of the multi-time equations \eqref{eq:multitimeeqwithsources}, $(\partial_{t_k} - s_k \partial_{z_k}) \psi^{(3)}_{s_1 s_2 s_3} = 0$, $k=1,2,3$. This leads to Cartesian products of three lines given by the characteristic variables $c_k:=z_k+s_kt_k$ appearing in the multi-time equations. For example, for $\psi^{(3)}_{+-+}(t_1,z_1;t_2,z_2;t_3,z_3)$ the multi-time characteristic is given by the set 
\begin{equation}
\left\{ \left. (s_1,y_1;s_2,y_2;s_3,y_3) \in \R^6 \ \right| \ s_1 + y_1 = c_1, s_2-y_2=c_2, s_3+y_3=c_3 \right\}.
\end{equation}
Figure \ref{fig:chars} shows two examples of multi-time characteristics with the three lines all drawn in one space-time diagram. 
\begin{figure}[h] \begin{center}\includegraphics[scale=0.85]{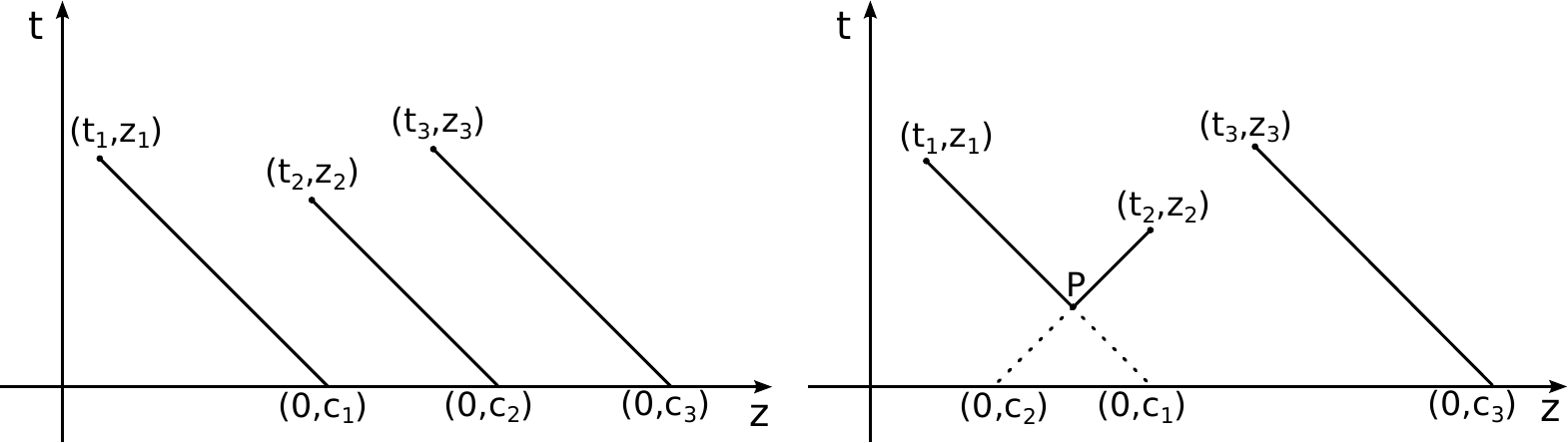}
\caption{Two examples for multi-time characteristics. The one on the left, for the component $\psi^{(3)}_{+++}$, does not intersect the coincidence point set $\coincidence$. The one the right, for $\psi^{(3)}_{+-+}$, intersects $\coincidence$ (cf. point $P$). } \label{fig:chars}
\end{center} \end{figure}
In the left picture, for the component $\psi^{(3)}_{+++}$, the lines of the multi-time characteristic can be followed back to the initial value surface without intersecting the boundary $\partial \spacelike^{(3)}$. ``Following back'' here means to choose a certain curve $\gamma(\tau)$ in the characteristic which connects $(t_1,z_1;t_2,z_2;t_3,z_3)$ with the point $(0,c_1;0,c_2;0,c_3)$.  Along the curve, the Dirac equation in the respective variables becomes an ordinary differential equation of the form $\frac{d}{d\tau} \psi(\gamma(\tau)) = f(\gamma(\tau))$ which can be integrated easily. Therefore, we will define an operator $I$ below that just integrates the inhomogeneity along the characteristic lines and gives the solution. We will choose our curve $\gamma(\tau)$ corresponding to a certain path in the space of the time variables along the three lines which comprise the characteristic. In our case:
\begin{equation}
(0,0,0) \longrightarrow (t_1,0,0) \longrightarrow (t_1,t_2,0)  \longrightarrow (t_1,t_2,t_3) .
\end{equation}
As it happens often in the study of multi-time equation, a change of this path in the time variables must not change the final result, which requires a certain integrability condition, called the \textit{consistency condition}. This condition was shown in lemma \ref{thm:nsectorsconsistency}.
In the right half of the picture, we additionally have to take the boundary condition into account because at the point $(P;P;t_3,z_3) \in \coincidence$, the characteristic intersects the coincidence point set $\coincidence$ and consequently leaves the domain. At this vertex point, the IBC has to be used, which we will implement via another operator $V$. The IBC then relates the value of the component $\psi^{(3)}_{s_1s_2 s_3}$ (here $\psi^{(3)}_{+-+}$) with the components of $\psi^{(2)}$ and with a different component of $\psi^{(3)}$, (here $\psi_{-++}^{(3)}$) which is associated with a different multi-time characteristic with one vertex less. One then follows this new multi-time characteristic back in time until either the boundary is reached again (then one repeats the process with a different component) or the initial surface $t_1 = t_2 = t_3 = 0$. In the picture on the right of Fig. \ref{fig:chars}, the IBC only has to be used once, as there is only one vertex point. In general, this procedure results in a formula where one uses the operators $I$ and $V$ alternatingly to obtain the solution from the initial data. The number of vertices in diagrams such as Figure \ref{fig:chars} determines how many steps the process takes.

\paragraph{Solution formula.} We construct the fixed point map sector-wise. To this end, let $n \in \N$ and assume that for all $k=1,...,n$, $j=1,...,n-1$, boundary functions $g^{(n)}_j \in D_{n-1}$, inhomogeneities $f^{(n)}_k \in D_n$ and initial values $\psi^{(n)}_0 \in C_b^1(Z_n, \C^{2^n})$ are given (see \eqref{eq:Zn} for the definition of $Z_n$). We shall solve the following initial boundary value problem: 
\begin{align} \label{eq:thegeneralsystem}
	i \left( \partial_{t_k} - s_k \partial_{z_k} \right) \psi^{(n)}_{s_1...s_n} ~&=~ f^{(n)}_{k,s_1...s_n}, 
	\nonumber\\ \left( \psi^{(n)}_{s_1...s_{j-1}-+s_{j+2}...s_n} - e^{i \theta} \psi^{(n)}_{s_1...s_{j-1}+-s_{j+2}...s_n} \right) & (x_1,...,x_j,x_j,x_{j+2},...,x_n)  \nonumber\\ & =~ g^{(n)}_{j,s_1...s_{j-1}s_{j+2}...s_n}(x_1,...,x_j,x_{j+2},...,x_n), \nonumber\\
	\ \psi^{(n)} |_{t_1=...=t_n=0} ~&=~ \psi^{(n)}_0.
\end{align}
Throughout the section we assume that the initial data are compatible with the boundary conditions in the sense of Eqs.\ \eqref{eq:IBCforinitialvaluesN} and \eqref{eq:IBCforinitialvaluesN2}.

The solution $\psi^{(n)}$ shall be constructed through repeated application of the operators $I$ and $V$ which we define now. For every $(t_1,...,t_n)\in \R^n $ and $t \in \R$, we define a map $I_t^{(t_1,...,t_n)}$ with ``$I$'' for ``integration along the characteristic''. 
\begin{equation} \begin{split} \label{eq:defI}
&I_t^{(t_1,...,t_n)}:  \ C_b^1(Z_n,\C^{2^n}) \to C_b^1(Z_n,\C^{2^n}),\\\
& \big(I_t^{(t_1,...,t_n)}\phi \big)_{s_1...s_n}(z_1,...,z_n)~=~\phi_{s_1...s_n}(c_1, ... , c_n)
\\ & -i \sum_{k=1}^n \int_{0}^{t_k} ds ~ f^{(n)}_{k,s_1...s_n}(t+t_1,z_1;...;t+t_{k-1},z_{k-1};t+s, c_k-s_k s; t, c_{k+1};...;t,c_n). \end{split}
\end{equation}
Here, $c_k=z_k + s_k t_k$ where $t_k$ is the time variable in the upper index of $I$ and $z_k$ the spatial variable from the argument of $I_t^{(t_1,...,t_n)}\phi$.

Moreover, for every $t \in \R, j \in \{1,...,n-1\}$ we define maps $V^{t}_j: C_b^1(Z_n,\C^{2^n}) \to C_b^1(Z_n,\C^{2^n})$ with ``$V$'' for ``switching indices at the vertex'' (at time $t$) by:
\begin{equation} \label{eq:defV}
\begin{array}{l c l} \left( V^{t}_j \phi \right)_{s_1...s_j s_j...s_n}(z_1,...,z_n) & = & \phi_{s_1...s_j s_j...s_n}(z_1,...,z_n), \vspace*{0.2cm}
\\  \left( V^t_j \phi \right)_{s_1...-+...s_n}(z_1,...,z_n) & = &
e^{i \theta} \phi_{s_1...+-...s_n}(z_1,...z_{j+1},z_j,...,z_n) \\ & & + \ g^{(n)}_{j,s_1...s_n} (t,z_1;...;\widehat{t,z_j};...;t,z_n),  \vspace*{0.2cm}
\\ \left( V^t_j \phi \right)_{s_1...+-...s_n}(z_1,...,z_n) & = &
e^{-i \theta} \big[ \phi_{s_1...-+...s_n}(z_1,...,z_{j+1},z_j,...,z_n)  \\ & &  - \ g^{(n)}_{j,s_1...s_n} (t,z_1;...;\widehat{t,z_j};...;t,z_n) \big]  ,
\end{array}
\end{equation}
where $\widehat{(\cdot)}$ denotes omission.

Furthermore, for every point $(t_1,z_1;...;t_n,z_n) \in \mathscr{S}_1^{(n)}$ and every spin index $s_1,...,s_n$, we define a set of \textit{collisions} according to the following rules. As previously, we let $c_k=z_k + s_kt_k$. A collision is a pair of indices $(j,k)$ in the set
\begin{equation}
\mathsf{Collisions} := \left\{ (j,k) \in \{1,...,n\}^2 :  j<k \ \text{but} \ c_j > c_k \right\}.
\end{equation}
$\mathsf{Collisions}$ is a finite set with $L:=|\mathsf{Collisions}|< n^2-1$. Its elements coincide with the index pairs of those lines which cross in the diagrams in Fig.\ \ref{fig:chars}. With each collision $(j,k)$, we associate a \textit{collision time}
\begin{equation}
\tau_{(j,k)} := \frac{1}{2}(c_j-c_k).
\end{equation}
We label these collision times in increasing order\footnote{There is a zero measure set of points for which this ordering is not possible since several collision times are equal. We omit these points in the upcoming considerations. Later, one can recover the value of the wave function at these points by continuation.}, i.e.\ identify each $\tau_{(j,k)} =: \tau_a$ with $a=1,...,L$ such that 
$\tau_1 < \tau_2 < ... < \tau_L$. Set $\tau_0 := 0$. Moreover, each collision is assigned an  index $k_a$ according to the formula:
\begin{equation} \label{eq:defparticleindex}
k_a((j,k)) := j + \left| \{ (j,l) \in \mathsf{Collisions} : l < k \} \right|, \quad a=1,...,L.
\end{equation}
In the diagrams in  Fig.\ \ref{fig:chars}, this number corresponds to the numbers of lines left of the vertex where the lines $j ,k$ cross plus one. The $\tau_a$ and $k_a$ are functions of the space-time point and the spin index only. Most importantly, they allow us to write down an explicit solution formula which is obtained by following the characteristics from collision to collision, as motivated heuristically above:
\begin{equation} \label{eq:solNIBC}
\psi^{(n)}_{s_1...s_n}(t_1,z_1;...;t_n,z_n) = \hspace*{-3pt} \left( I_{\tau_L}^{(t_1-\tau_L,...,t_n-\tau_L)} \left( \prod_{a=1}^L V^{\tau_a}_{k_a} I_{\tau_{a-1}}^{(\tau_a - \tau_{a-1},...,\tau_a - \tau_{a-1})} \right) \psi^{(n)}_0 \right) _{s_1...s_n} \hspace*{-0.9cm} (z_1,...,z_n).
\end{equation}
It is understood that the factors in the product are written from right to left, i.e.:
\begin{equation}
\left( I_{\tau_L}^{(t_1-\tau_L,...,t_n-\tau_L)} V^{\tau_L}_{k_L} I^{(\tau_L - \tau_{L-1},...,\tau_L - \tau_{L-1})} \dots V^{\tau_1}_{k_1} I_0^{(\tau_1,...,\tau_1)} \psi^{(n)}_0 \right)_{s_1...s_n} (z_1,...,z_n).
\end{equation}

\begin{lemma} \label{thm:existencinsectorn} 
Let $n \in \N$. The function $\psi^{(n)}$ defined by \eqref{eq:solNIBC} for given $f^{(n)}_k$ and $g^{(n)}_j$ is the unique solution of the IBC system \eqref{eq:thegeneralsystem} in the $n$-th sector,
provided the inhomogeneities satisfy, for all $j \neq k$, the consistency conditions
\begin{equation} \label{eq:ccgeneraln}
i\left( \partial_{t_k} - s_k \partial_{z_k} \right) f_{j, s_1...s_n}^{(n)} ~=~ i\left( \partial_{t_j} - s_j \partial_{z_j} \right) f_{k, s_1...s_n}^{(n)} . 
\end{equation}
\end{lemma}

\begin{proof}
For $L \in \N_0$, we prove via induction over the number $L$ of collisions the statement $\mathbf{A(L)}$: \emph{A function $\psi^{(n)} \in \mathcal{D}$ solves the IBC system \eqref{eq:thegeneralsystem} at all points $(t_1,z_1;...;t_n,z_n)$ and for all spin indices $s_1,...,s_n$ for which $|\mathsf{Collisions}| \leq L$ if and only if it is given by \eqref{eq:solNIBC} at those points.}
\\[0.2cm]
\textbf{Base Case} $\mathbf{A(0)}$. At points with $|\mathsf{Collisions}|=0$, \eqref{eq:solNIBC} yields
\begin{align} \label{eq:nocollision}
\psi^{(n)}_{s_1...s_n}&(t_1,z_1;...; t_n,z_n) ~=~ \psi^{(n)}_{0,s_1...s_n}(c_1,...,c_n) 
\nonumber\\ -i &  \sum_{k=1}^n \int_{0}^{t_k} ds ~ f_{k,s_1...s_n}^{(n)}(t_1,z_1;...;t_{k-1},z_{k-1};s, c_k-s_ks; 0, c_{k+1};...;0,c_n).
\end{align}
We first show that \eqref{eq:nocollision} indeed is a solution of the IBC system. As $L=0$,  the IBC does not come into play here. The initial conditions are satisfied by construction (see \eqref{eq:solNIBC}). We now calculate the derivatives w.r.t. the $n$-th coordinates. Omitting spin indices, we find:
\begin{align} \label{eq:calculforn}
i (\partial_{t_n} - s_n \partial_{z_n} ) \psi^{(n)} ~&=~   f_n^{(n)}(t_1,z_1;...;t_n, c_n-s_nt_n)  + \int_0^{t_n} \hspace*{-0.1cm} ds \underbrace{i (\partial_{t_n} - s_n\partial_{z_n} ) f_{k}^{(n)}(...;s,c_n-s_n s)}_{=0}
\nonumber\\ &~~~ + \sum_{k=1}^{n-1} \int_0^{t_k} \hspace*{-0.1cm} ds \underbrace{i (\partial_{t_n} - s_n\partial_{z_n} ) f_{k}^{(n)}(...;0,c_n)}_{=0} 
\nonumber\\ &=~ f_n^{(n)}(t_1,z_1;...;t_n, z_n).
\end{align}
Next, we consider an arbitrary particle index $j \neq n$. Because \eqref{eq:nocollision} is not symmetric in the particle indices, more work is required to see that the multi-time equation is satisfied. We shall use the consistency condition \eqref{eq:ccgeneraln} to show this. We compute via a telescoping sum:
\begin{equation} \label{eq:telescopesum}
\begin{split} 
\int_0^{t_j} ds \, & f_j^{(n)}(t_1,z_1;...;t_{j-1},z_{j-1};s,c_j-s_js;0,c_{j+1};...;0,c_n)
\\  =  \int_0^{t_j} ds \,& f_j^{(n)}(t_1,z_1;...;t_{j-1},z_{j-1};s,c_j-s_js;t_{j+1},z_{j+1};...;t_n,z_n)
\\  + \sum_{l=j+1}^n  \int_0^{t_j} ds \,  \Big( & f_j^{(n)}(t_1,z_1;...;t_{j-1},z_{j-1};s,c_j-s_js;0,c_{j+1};...;0,c_{l-1};0,c_l;t_{l+1},z_{l+1};...;t_n,z_n)  \\    -  & f_j^{(n)}(t_1,z_1;...;t_{j-1},z_{j-1};s,c_j-s_js;0,c_{j+1};...;0,c_{l-1};t_l,z_l;t_{l+1},z_{l+1};...;t_n,z_n) \Big).
\end{split}
\end{equation}
Next, we use the fundamental theorem of calculus, abbreviating $(t_1,z_1; ...;t_{j-1},z_{j-1})$ by $\star$ and $(t_{l+1},z_{l+1};...;t_n,z_n)$ by $\sharp$,
\begin{equation} \label{eq:fcalculation} \begin{split}
\int_0^{t_j} ds  \Big( & f_j^{(n)}(\star; s,c_j-s_js;0,c_{j+1};...;0,c_{l-1};0,c_l; \sharp) \\    - & f_j^{(n)}(\star ;s,c_j-s_js;0,c_{j+1};...;0,c_{l-1};t_l,z_l;\sharp) \Big) \\ = - \int _0^{t_j} ds \int_0^{t_l} dr & \, \frac{d}{dr} f_j^{(n)}(\star; s,c_j-s_js;0,c_{j+1};...;0,c_{l-1};r,c_l-s_lr;\sharp) 
\\ \stackrel{\eqref{eq:ccgeneraln} }{=} - \int_0^{t_l} dr  \int _0^{t_j} ds & \, \frac{d}{ds} f_l^{(n)}(\star;s,c_j-s_js;0,c_{j+1};...;0,c_{l-1};r,c_l-s_lr;\sharp) 
\\ = \int_0^{t_l} dr \Big( &  f_l^{(n)}(\star;0,c_j;0,c_{j+1};...;0,c_{l-1};r,c_l-s_lr;\sharp)  \\  - & f_l^{(n)}(\star;t_j,z_j;0,c_{j+1};...;0,c_{l-1};r,c_l-s_lr;\sharp) \Big),
 \end{split}
\end{equation}
 where the consistency condition \eqref{eq:ccgeneraln} has been used to obtain $\tfrac{d}{dr}f_j^{(n)}(\cdots) = \tfrac{d}{ds} f_l^{(n)}(\cdots)$ with the argument $(\cdots)$ as in \eqref{eq:fcalculation}. Inserting the result of \eqref{eq:fcalculation} in all the summands in \eqref{eq:telescopesum} leads to
\begin{equation} \begin{split}
\int_0^{t_j} ds \, & f_j^{(n)}(\star;s,c_j-s_js;0,c_{j+1};...;0,c_n)
\\ =  \int_0^{t_j} ds \, & f_j^{(n)}(\star;s,c_j-s_js;t_{j+1},z_{j+1};...;t_n,z_n)
\\ + \sum_{l=j+1}^n \int_0^{t_l} ds \, \Big( &  f_l^{(n)}(\star;0,c_j;0,c_{j+1};...;0,c_{l-1};s,c_l-s_l s;\sharp)  \\  - & f_l^{(n)}(\star;t_j,z_j;0,c_{j+1};...;0,c_{l-1};s,c_l-s_l s;\sharp) \Big).
\end{split}
\end{equation} 
Inserting this expression into \eqref{eq:nocollision} yields: 
\begin{equation}
\begin{split}
\psi^{(n)}_{s_1...s_n} & (t_1,z_1;...;t_n,z_n) =  \psi^{(n)}_0(c_1,...,c_n) 
\\ -i & \sum_{k < j} \int_{0}^{t_k} ds ~ f_{k,s_1...s_n}^{(n)}(t_1,z_1;...;t_{k-1},z_{k-1};s, c_k-s_ks; 0, c_{k+1};...;0,c_n)
\\ -i & \int_{0}^{t_j} ds ~ f_{j,s_1...s_n}^{(n)}(t_1,z_1;...;t_{k-1},z_{k-1};s, c_j-s_js; t_{k+1}, z_{k+1};...;t_n,z_n)
\\ -i & \sum_{k>j} \int_{0}^{t_k} ds ~ f_{k,s_1...s_n}^{(n)}(t_1,z_1;...;0,c_j;...;t_{k-1},z_{k-1};s, c_k-s_ks; 0, c_{k+1};...;0,c_n). 
\end{split}
\end{equation} 
Considering this expression, it becomes obvious that the same calculation as in \eqref{eq:calculforn} results in
\begin{equation}
i (\partial_{t_j} - s_j \partial_{z_j} ) \psi^{(n)}_{s_1...s_n}(t_1,z_1,...,t_n,z_n)  = f_{j,s_1...s_n}^{(n)}(t_1,z_1,...,t_n, z_n),
\end{equation} 
as desired. So \eqref{eq:nocollision} defines a solution of the IBC system for points without collisions. 

The fact that \eqref{eq:nocollision} gives the only solution of the IBC
system on the points under consideration follows from the uniqueness of
solutions of each single equation $(\partial_{t_k} - s_k \partial_{z_k})
\psi^{(n)}= f^{(n)}_k$ )see \cite{evans} and compare also \cite[thm.\ 4.4]{LN:2015}).
\\[0.2cm]
\textbf{Induction step} $\mathbf{A(L-1)} \Rightarrow \mathbf{A(L)}$. Let a point $(t_1,z_1,...,t_n,z_n)$ and spin indices $(s_1,...,s_n)$ be given such that $|\mathsf{Collisions}|=L$. The collision with the greatest time $t_L$ must have the form $(k_L, k_{L}+1)$ with $k_L$ given by \eqref{eq:defparticleindex}. (This can be seen from diagrams such as Fig.\ \ref{fig:chars} and is easy to prove using the claim in the proof of \cite[lemma 6.2]{LN:2015}.) W.l.o.g. we assume that $s_{k_L} = +1$ and $s_{k_{L}+1}=-1$. The reversed case is the only other possible one; it can be treated analogously.

We proceed in two steps: First we connect the value of $\psi^{(n)}$ at $(t_1,z_1,...,t_N,z_N)$ with the value at the largest collision time $\tau_L$ via integration along the multi-time characteristics. Secondly, we implement the IBC at time $\tau_L$ via the operator $V$ acting on a wave function component with only $L-1$ collisions (that is known due to the induction assumption).

For the first step, suppose the function $\psi^{(n)}_{\tau_L} = \psi^{(n)}|_{t_1=...=t_n=\tau_L} \in C_b^1(Z_n, \C^{2^n})$ is given. Analogously to the base case $L=0$, it follows that the component $\psi^{(n)}_{s_1...s_n}$ solves the equations
\begin{equation}
i \left( \partial_{t_k} - s_k \partial_{z_k} \right) \psi^{(n)}_{s_1...s_n} ~=~ f_{k,s_1...s_n}^{(n)}, \quad k=1,...,n
\end{equation}
if and only if it is given by
\begin{equation}
\psi^{(n)}_{s_1...s_n}(t_1,z_1,...,t_n,z_n) ~=~ \left( I_{\tau_L}^{(t_1-\tau_L,...,t_n-\tau_L)}\psi^{(n)}_{\tau_L} \right)_{s_1...s_n} (z_1,...,z_n).
\end{equation} 
In the second step, we want to find $\psi^{(n)}_{\tau_L}$. The multi-time characteristic associated with $\psi^{(n)}_{s_1...s_n}(t_1,z_1,...,t_n,z_n)$ intersects the boundary of $\mathscr{S}_1^{(n)}$ in $P:=(\tau_L, z_1 + s_1 (\tau_L - t_1), ..., \tau_L,$ $z_n + s_n (\tau_L - t_n))$.
At that point, the component $\psi^{(n)}_{s_1...-+...s_n}$ (where $- = s_{k_L}$ and $+ = s_{k_{L+1}}$) has one collision less than $\psi^{(n)}_{s_1...+-...s_n}$, so it has $L-1$ collisions with the same times $\tau_a$ and indices $k_a$, $a=1,...,L-1$. By the induction hypothesis, it is then given by \eqref{eq:solNIBC}, i.e.
\begin{equation} \begin{split}
& \psi^{(n)}_{s_1...-+...s_n}(P) =  \bigg( I_{\tau_{L-1}}^{(\tau_L-\tau_{L-1},...,\tau_L-\tau_{L-1})}  \times \\ & \times \bigg( \prod_{a=1}^{L-1} V^{\tau_a}_{k_a} I_{\tau_{a-1}}^{(\tau_a - \tau_{a-1},...,\tau_a - \tau_{a-1})} \bigg) \psi^{(n)}_0 \bigg) _{s_1...-+...s_n}(z_1 + s_1 (\tau_L - t_1),...,z_n + s_n (\tau_L - t_n)). \end{split}
\end{equation}
By comparison of \eqref{eq:defV} with the IBC from \eqref{eq:thegeneralsystem}, it becomes apparent that the latter evaluated at $P$ is equivalent to
\begin{equation}
\psi^{(n)}_{s_1...+-...s_n} (P) = \left( V_{k_L}^{\tau_L} \psi^{(n)}(\tau_L, \cdot, ..., \tau_L, \cdot) \right)_{s_1...+-...s_n} (z_1 + s_1 (\tau_L - t_1),...,z_n + s_n (\tau_L - t_n)).
\end{equation}
Therefore, combining both steps, we see that the system \eqref{eq:thegeneralsystem} is satisfied if and only if $\psi^{(n)}_{s_1...+-...s_n}(t_1,z_1,...,t_n,z_n)$ is given by \eqref{eq:solNIBC}. 
This finishes the induction and thus the proof. \qed
\end{proof}

The insight that Eq. \eqref{eq:solNIBC} gives the solution in a specified sector is the basis of the following central definition. 

\begin{definition}[Definition (fixed point map).] Recall the definition \eqref{eq:banachn} of the spaces $\Banach_n$ and let
\begin{equation}
\mathcal{D}:= \left\lbrace (\psii, ..., \psi^{(N)}) \in \Banach_1 \oplus \dots \oplus \Banach_N : \psi^{(n)}\big|_{t_1=...=t_n=0}=\psi^{(n)}_0  \ \forall n =1,...,N \right\rbrace.
\end{equation}
Then the map $F: \mathcal{D} \to \mathcal{D}$, $(v^{(1)}, ..., v^{(N)}) \mapsto (\psi^{(1)},...,\psi^{(N)})$ is defined by the following procedure:
\begin{itemize}
\item Let $\psi^{(1)}$ be given by formula \eqref{eq:solNIBC} for $n=1$ with no boundary terms ($g^{(1)} = 0$ since $\partial \mathscr{S}^{(1)} = \emptyset$) and the inhomogeneity 
	\begin{equation}
f^{(1)}_{1,s}(x_1) ~= -\sum_{t,u=\pm 1} A^{tu}_{s} v^{(2)}_{tu}(x_1,x_1).
	\end{equation}
\item Repeat the following for all $n=2,...,N-1$ in ascending order: $\psi^{(n)}$ is defined by formula \eqref{eq:solNIBC} with boundary terms given, as in \eqref{eq:mgeneralibcs}, by the already determined function $\psi^{(n-1)}$ and the inhomogeneity
	\begin{equation}
f_{k,s_1...s_n}^{(n)}(x_1,...,x_n) ~= \sum_{t,u = \pm 1} (-1)^k A_{s_k}^{tu} v^{(n+1)}_{s_1...s_{k-1}\,  t \, u \, s_{k+1}...s_n}(x_1,...,x_k,x_k,x_{k+1},...,x_n)	.
	\end{equation}	

\item Finally, $\psi^{(N)}$ is defined by \eqref{eq:solNIBC} with boundary terms given, as in \eqref{eq:mgeneralibcs}, by the already determined $\psi^{(N-1)}$ and the inhomogeneity $f^{(N)}=0$.
\end{itemize}
\end{definition}

\paragraph{Well-definedness of $F$.} To show that $F$ actually maps into $\mathcal{D}$, we have to check the $C^1_b$-property of the function defined by \eqref{eq:solNIBC}. Since for $n=0$, the set $\mathsf{Collisions}$ is empty, $\psii \in C^1_b$ follows directly from the properties of the initial values. For $n \geq 2$, one has to consider those points separately where for some $j<k, c_j=c_k$. This occurs in $\mathscr{S}_1^{(n)}$ for positive times only if $s_j=+1$ and $s_k=-1$.
When $c_j \searrow c_k$, the collision time $\tau_{(j,k)}$ approaches $0$, so continuity at a point with $c_j = c_k$ amounts to:
\begin{equation}
\psi^{(n)}_{0,s_1...+-...s_n}(z_1,...z_{j-1},z,z,...,z_n) ~=~ V^0_{j} \psi^{(n)}_{0,s_1...+-...s_n}(z_1,...z_{j-1},z,z,...,z_n),
\end{equation}
which follows from \eqref{eq:IBCforinitialvaluesN}. By an argument analogous to the one given in \eqref{eq:2sectorC1ofF} for $N=2$, one can see that the $z$-derivative of \eqref{eq:IBCforinitialvaluesN} together with \eqref{eq:IBCforinitialvaluesN2} implies the $C_b^1$-property of the functions $\psi^{(n)}$.

To continue with the fixed point argument, we endow the spaces $\Banach_n$ with weighted norms, similarly as in \eqref{eq:2gammanorms}. For $\gamma \geq 0$, 
\begin{equation} \label{eq:defgammanorms} \begin{split}
\big\| f \big\|_{\Banach_n,\gamma} := & \sup_{\substack{t_1,...,t_n \in [0,T] \\ (z_1,...,z_n) \in Z_n}} \bigg( \Big( |f(t_1,z_1,...t_n,z_n)| \,  +
\\ & ~+ \max_{y \in \{t_1,z_1,...t_n,z_n \}} |\partial_y f(t_1,z_1,...t_n,z_n)| \Big) e^{-\tfrac{\gamma}{n}(t_1 +...+t_n)} \bigg). \end{split}
\end{equation}
As in the case $N=2$, $| \cdot |$ denotes the maximum norm in the finite dimensional spaces $\C^{2^n}$. Moreover, the norm on $C^1_b$ is defined as 
\begin{equation} \label{eq:defc1bnorms}
\big\| f \big\|_{C^1_b(Z_n,\C^{2^n})} :=  \sup_{(z_1,...,z_n) \in Z_n}  \Big( |f(z_1,...,z_n)| \  
~+ \max_{y \in \{z_1,...,z_n \}} |\partial_y f(z_1,...,z_n)| \Big). 
\end{equation}
 The constructive proof of the previous lemma directly yields a bound of the norm of $\psi^{(n)}$.

\begin{lemma} 
The function $\psi^{(n)}$ given by \eqref{eq:solNIBC} satisfies the bound
\begin{equation} \label{eq:boundonpsin}
\left\| \psi^{(n)} \right\|_{\Banach_n,\gamma} \leq ~\left\| \psi_0 \right\|_{C_b^1(Z_n,\C^{2^n})} + \frac{2n^4}{\gamma} \max_{1 \leq k \leq n} \left\| f^{(n)}_k \right\|_{\Banach_n,\gamma} + n^2 \max_{1 \leq j \leq n-1} \left\| g^{(n)}_j \right\|_{\Banach_{n-1},\gamma}.
\end{equation}
\end{lemma}

\begin{proof}
Let $\phi \in C_b^1(Z_n,\C^{2^n})$. Then by \eqref{eq:defI},
\begin{equation}
\begin{split}
& \left| I_0^{(t_1,...,t_n)}\phi(z_1,...,z_n) \right| -  \left| \phi(c_1,...,c_n) \right|  \leq  \sum_{k=1}^{n}\left|  \int_0^{t_k} ds \ f^{(n)}_{k}(t_1,z_1;...;s,c_k-s_ks;0,c_{k+1};...) \right|
\\ & \leq n \max_{k \in \{1,...,n\}} \left( \int_0^{t_k} ds \ \big| f^{(n)}_{k}(t_1,z_1;...;s,c_k-s_ks;0,c_{k+1};...)\big| e^{-\tfrac{\gamma}{n} s} e^{\tfrac{\gamma}{n} s} \right).
\\ & ~\leq n \max_{k \in \{1,...,n\}} \sup_{s \in [0,t_k]} \left| f^{(n)}_k(t_1,z_1;...;s,c_k-s_ks;0,c_{k+1};...) e^{-\tfrac{\gamma}{n} s} \right| \int_0^{t_k} ds ~ e^{\tfrac{\gamma}{n} s}
\\ & \leq n \max_{k \in \{1,...,n\}} \sup_{s \in [0,t_k]} \left| f^{(n)}_k(t_1,z_1;...;s,c_k-s_ks;0,c_{k+1};...) e^{-\tfrac{\gamma}{n} (s+t_1+...+t_{k-1})} \right| \frac{n}{\gamma} e^{\tfrac{\gamma}{n}(t_k + t_1 + ... + t_{k-1})}.
\end{split}
\end{equation}
This implies:
\begin{equation} \begin{split}
& \left| I_0^{(t_1,...,t_n)}\phi(z_1,...,z_n) \right| e^{-\tfrac{\gamma}{n} (t_1+...+t_{n})} ~\leq~  \left| \phi(c_1,...,c_n) \right|  e^{-\tfrac{\gamma}{n} (t_1+...+t_{n})}+ \\ &+ \frac{n^2}{\gamma} \max_{k \in \{1,...,n\}} \sup_{s \in [0,t_k]} \big| f^{(n)}_k(t_1,z_1;...;s,c_k-s_ks;0,c_{k+1};...)\big| e^{-\tfrac{\gamma}{n} (s+t_1+...+t_{k-1})} . \end{split}
\end{equation}
Changing the starting time from 0 to $t$ changes only little,
\begin{equation} \begin{split}
& \left| I_t^{(t_1,...,t_n)}\phi(z_1,...,z_n) \right| e^{-\tfrac{\gamma}{n} (t+ t_1+...+t+ t_{n})} ~\leq ~  \left| \phi(c_1,...,c_n) \right|  e^{-\tfrac{\gamma}{n} (t+t_1+...+t+t_{n})}+ \\ &+ \frac{n^2}{\gamma} \max_{k \in \{1,...,n\}} \sup_{s \in [0,t_k]} \big| f^{(n)}_k(t+t_1,z_1;...;t+s,c_k-s_ks;t,c_{k+1};...)\big|\, e^{-\tfrac{\gamma}{n} (t+s+t+t_1+...+t+t_{k-1})} . \end{split} \end{equation}
Hence, together with a similar consideration for the derivatives (analogous to the one for $N=2$, $n=1$ in \eqref{eq:2boundderivatives}) we find:
\begin{equation} \label{eq:boundonI}
 \left\| I_t^{(t_1,...,t_n)}\phi \right\|_{\Banach_n,\gamma} ~\leq~  \left\| \phi \right\|_{C_b^1(Z_n, \C^{2^n})} + \frac{2n^2}{\gamma} \max_{1 \leq k \leq n} \left\| f^{(n)}_k \right\|_{\Banach_n,\gamma}.
\end{equation}
Similarly, it follows from the definition \eqref{eq:defV} that
\begin{equation}
\left\| V^t_j \phi \right\|_{C_b^1(Z_n, \C^{2^n})} e^{-\tfrac{\gamma}{n}nt} ~\leq~ \left\| \phi \right\|_{C_b^1(Z_n, \C^{2^n})}e^{-\tfrac{\gamma}{n} nt} + \left\| g^{(n)}_j \right\|_{C_b^1(Z_{n-1}, \C^{2^{n-1}})}e^{-\tfrac{\gamma}{n-1}(n-1) t}.
\end{equation}
We see that each application of $V^t_k$ leads to an additive contribution of at most
\begin{equation}
\max_{1 \leq j \leq n} \left\| g^{(n)}_j \right\|_{C_b^1(Z_{n-1}, \C^{2^{(n-1)}})} \times e^{-\tfrac{\gamma}{n-1}(n-1) t},
\end{equation}
where the last factor is the appropriate weight factor of the $\gamma$-norm. Thus,
	\begin{equation} \label{eq:boundonV}
\left\| V^t_k \phi \right\|_{\Banach_n, \gamma} ~\leq~ \left\| \phi \right\|_{C_b^1(Z_n, \C^{2^n})} + \max_{1 \leq j \leq n-1} \left\| g^{(n)}_j \right\|_{\Banach_{n-1},\gamma}.
	\end{equation}
We know that $\psi^{(n)}$ is given by the formula \eqref{eq:solNIBC}. If there are at most $L$ collisions, the operators $I$ and $V$ are applied at most $L+1$ times. Each time, the terms from equations \eqref{eq:boundonI} and \eqref{eq:boundonV} add up.  Therefore, we obtain the following bound of the norm of $\psi^{(n)}$:
	\begin{equation}
\left\| \psi^{(n)} \right\|_{\Banach_n,\gamma}  \leq \left\| \psi^{(n)}_0 \right\|_{C_b^1(Z_n,\C^{2^n})} +  \left( L+1 \right) \frac{2n^2}{\gamma}  \max_{1 \leq k \leq n} \left\| f^{(n)}_k \right\|_{\Banach_n,\gamma} + (L+1) \max_{1 \leq j \leq n-1} \left\| g^{(n)}_j \right\|_{\Banach_{n-1},\gamma}.
	\end{equation}
Since the number of collisions is bounded by $L+1 < n^2$, this yields \eqref{eq:boundonpsin}. \qed
\end{proof}

\begin{lemma} \label{thm:mlemmaF} Let $(\psi^{(1)},...,\psi^{(N)}) \in \Banach_1 \oplus \dots \oplus \Banach_N$. Then the following statements are equivalent:
\begin{enumerate}
\item[(i)] $(\psi^{(1)}, ..., \psi^{(N)})$ is a $C_b^1$-solution of the IBC system \eqref{eq:multitimeeqwithsources}--\eqref{eq:mgeneralibcs} with initial values given by \eqref{eq:minitial}.
\item[(ii)] $(\psi^{(1)},..., \psi^{(N)}) \in \mathcal{D}$ is a fixed point of $F$.
\end{enumerate} 
\end{lemma} 

\begin{proof}
On the one hand, if (i) holds then we have $(\psi^{(1)}, ..., \psi^{(N)}) \in \mathcal{D}$ and $F$ is constructed such that it does not change the functions $\psi^{(n)}$.

On the other hand, let $(\psi^{(1)},..., \psi^{(N)}) \in \mathcal{D}$ be a fixed point of $F$. Then the initial conditions are satisfied by definition of $\mathcal{D}$. It remains to check that $(\psi^{(1)},..., \psi^{(N)})$ also satisfies the multi-time equations and IBCs. To demonstrate this, we would like to apply lemma \ref{thm:existencinsectorn}. However, in order to show that the IBC system is satisfied for some sector $1 \leq n \leq N$, lemma \ref{thm:existencinsectorn} requires the consistency conditions \eqref{eq:ccgeneraln} for that $n$. Now for $n=N$, \eqref{eq:ccgeneraln} is satisfied trivially as $f_k^{(N)} = 0 \, \forall k$. That means, lemma \ref{thm:existencinsectorn} yields that $\psi^{(N)}$ satisfies the IBC system for the $N$-th sector. Now we can use lemma \ref{thm:nsectorsconsistency} to conclude that the consistency condition \eqref{eq:ccgeneraln} holds also for $n=N-1$. Thus, lemma \ref{thm:existencinsectorn} shows that the IBC system is satisfied for $n=N-1$. Using the two lemmas alternatingly for descending $n = N-2,...,1$ eventually yields (i).
\qed
\end{proof}

\begin{lemma}
For every $T>0$, $F$ has a unique fixed point in $\mathcal{D}$.
\end{lemma}

\begin{proof}
We shall use Banach's fixed point theorem. To this end, we first show that $\mathcal{D}$ is a closed subset of a Banach space. For $\gamma =0$, the weighted norms \eqref{eq:defgammanorms} on $\Banach_n$ are the canonical norms on $C^1_b$-function, and it is well-know that this leads to complete spaces. Since we have the inequalities
\begin{equation}
\left\| f^{(n)} \right\|_{\Banach_n, 0} e^{-\gamma T} \leq \left\| f^{(n)} \right\|_{\Banach_n, \gamma}\leq \left\| f^{(n)} \right\|_{\Banach_n, 0} \quad \forall n \in \N,
\end{equation}
the $\| \cdot \|_{\Banach_n, \gamma}$ norms are all equivalent for all $\gamma \geq 0$. Thus, $\Banach_1 \oplus \dots \oplus \Banach_N$ equipped with the norm
\begin{equation}
	\left\| (f^{(1)},...,f^{(N)}) \right\|_{\Banach_1 \oplus \dots \oplus \Banach_N, \gamma} = ~	\sum_{n=1}^N \left\| f^{(n)} \right\|_{\Banach_n,\gamma}.
\end{equation}
is a Banach space. Moreover, it is easy to see that $\mathcal{D} \subset \Banach_1 \oplus \dots \oplus \Banach_N$ is a closed subset. 
\\ It remains to show that $F: \mathcal{D} \to \mathcal{D}$ is a contraction. To this end, let 
\begin{equation} \begin{split}
& (v^{(1)},...,v^{(N)}), (w^{(1)},...,w^{(N)}) \in  \mathcal{D}, \\ & (\psi^{(1)},...,\psi^{(N}) :=F(v^{(1)},...,v^{(N)}), (\phi^{(1)},...,\phi^{(N)}) :=F(w^{(1)},...,,w^{(N)}).
\end{split}
\end{equation}
By linearity of the IBC system, $\psi^{(n)} - \phi^{(n)}$ is, for $n=1,...,N$, a solution of \eqref{eq:thegeneralsystem} with the inhomogeneities and boundary values replaced by the difference between the ones for $\psi^{(n)}$ and $\phi^{(n)}$ and initial data equal to zero. 
Keeping this in mind, we now prove the following (crude) estimate by induction over $n \leq N$: 
\begin{equation} \label{eq:boundforfixedpoint} 
\left\| \psi^{(n)} - \phi^{(n)} \right\|_{\Banach_n, \gamma}~ \leq~ \frac{1}{\gamma} 
(2N)^{4n}
\left\| A\right\|_\infty \left( \sum_{k=2}^{\min\{n+1, N\}} \left\| B \right\|_{\infty}^{n+1-k} \left\| v^{(k)} - w^{(k)}\right\|_{\Banach_k,\gamma} \right).
\end{equation}
\textbf{Base case $n=1$:} Considering that for $n=1$, there are no boundary conditions, the bound \eqref{eq:boundonpsin} directly leads to:
\begin{equation} \begin{split}
\left\| \psi^{(1)} - \phi^{(1)} \right\|_{\Banach_1, \gamma} ~& \leq~ 0 + \frac{2^4}{\gamma} \left\| A \, \big(v^{(2)}(x_1,x_1) - w^{(2)}(x_1,x_1) \big) \right\|_{\Banach_1,\gamma}  
\\ & \leq~ \frac{1}{\gamma} (2N)^4 \left\| A \right\|_\infty \left\| v^{(2)} - w^{(2)}\right\|_{\Banach_2,\gamma}, 
\end{split}
\end{equation}
i.e., \eqref{eq:boundforfixedpoint} for $n=1$.\\[0.2cm]
\textbf{Induction step $n \rightarrow n+1$:} Assume that \eqref{eq:boundforfixedpoint} holds for some $n \geq 1$. We first consider the case $n \leq N-2$. Using \eqref{eq:boundonpsin} for $n+1$ and then the induction hypothesis yields:
\begin{equation} \begin{split}
\left\| \psi^{(n+1)} - \phi^{(n+1)} \right\|_{\Banach_{n+1}, \gamma} & \leq~ \frac{(2N)^4}{\gamma} \left\| A \right\|_\infty \left\| v^{(n+2)} - w^{(n+2)} \right\|_{\Banach_{n+2,\gamma}} \\ &~~~~ + N^2 \left\| B \right\|_\infty \left\| \psi^{(n)} - \phi^{(n)} \right\|_{\Banach_{n}, \gamma}
\\ & \leq~ \frac{1}{\gamma} 
 (2N)^{4(n+1)} \left\| A\right\|_\infty \left( \sum_{k=2}^{n+2} \left\| B \right\|_{\infty}^{n+2-k} \left\| v^{(k)} - w^{(k)}\right\|_{\Banach_k,\gamma} \right).
 \end{split}
\end{equation}
In case $n= N-1$, the same calculation goes through, except that the term with $v^{(n+2)}-w^{(n+2)}$ is absent since $f_k^{(N)} = 0 \, \forall k$. Thus, the sum ends at $N$. This proves \eqref{eq:boundforfixedpoint}.

Next, we determine an upper bound on the total norm, 
\begin{align} \label{eq:boundonFixed}
&\left\| \big(\psi^{(1)},...,\psi^{(N)}\big)-\big(\phi^{(1)},...,\phi^{(N)}\big) \right\|_{\Banach_1 \oplus \dots \oplus \Banach_N, \gamma} ~\leq~ \sum_{n=1}^{N} \left\| \psi^{(n)} - \phi^{(n)} \right\|_{\Banach_n, \gamma} \nonumber\\
&\leq~ \frac{1}{\gamma} 
(2N)^{4N}  \left\| A \right\|_\infty \sum_{n=1}^N \sum_{k=2}^{\min \{n+1,N\}} \left\| B \right\|_{\infty}^{n+1-k} \left\| v^{(k)} - w^{(k)}\right\|_{\Banach_k,\gamma}\nonumber\\ 
 &\leq~ \frac{1}{\gamma} 
(2N)^{4N}  \left\| A \right\|_\infty \sum_{n=1}^N ~ \max \{1, \left\| B \right\|_{\infty}^{N} \} \sum_{k=1}^{N} \left\| v^{(k)} - w^{(k)}\right\|_{\Banach_k,\gamma}\nonumber\\
 &\leq~ \frac{1}{\gamma} 
(2N)^{4N+1} \left\| A \right\|_\infty \max \{1, \left\| B \right\|_{\infty}^{N} \} \left\| (v^{(1)},...,v^{(N)}) - (w^{(1)},...,w^{(N)}) \right\|_{\Banach_1 \oplus \dots \oplus \Banach_N, \gamma}.
\end{align}
If we now choose, for example,  $\gamma=5 \,(2N)^{4N+1} \left\| A\right\|_\infty \max \{1, \left\| B \right\|_{\infty}^{N} \}$, then \eqref{eq:boundonFixed} yields:
\begin{equation} \begin{split}
&\left\| F(v^{(1)},...,v^{(N)}) - F(w^{(1)},...,w^{(N)}) \right\|_{\Banach_1 \oplus \dots \oplus \Banach_N, \gamma} \\ \leq~ \frac{1}{5} &\left\| (v^{(1)},...,v^{(N)}) - (w^{(1)},...,w^{(N)}) \right\|_{\Banach_1 \oplus \dots \oplus \Banach_N, \gamma}.
\end{split}
\end{equation}
Thus $F$ is a contraction, and Banach's fixed point theorem yields the claim. \qed
\end{proof}

Together with lemma \ref{thm:mlemmaF}, this finishes the proof of theorem \ref{thm:mmain}, establishing that the multi-time IBC system has a unique $C_b^1$-solution for all positive times.

\subsection{Proof of theorem \ref{thm:nibc}} \label{sec:proofnibc}

	We need to show that \eqref{eq:currentcondition} holds. To this end, we consider both sides of the condition separately. For $n=N$, we just have the free evolution equations on $\spacelike_1^{(N)}$, which ensure \eqref{eq:currentcondition}. We now turn to $1 \leq n <N$.

\noindent On the one hand, we have (writing out $j^{\mu_1 ... \mu_{n+1}}$ in components):
	\be
		j^{\mu_1 ... \mu_{n+1}}~ = \sum_{s_1 ... s_{n+1} = \pm 1} |\psi^{(n+1)}_{s_1 ... s_{n+1}}|^2 \, (-s_1)^{\mu_1} \cdots (-s_{n+1})^{\mu_{n+1}},
		\label{eq:currentincomponents}
	\ee
and therefore:
\begin{align}
	&j^{\mu_1 ... \mu_{k-1}\, 0 1 \,\mu_{k+1} ... \mu_n} - j^{\mu_1 ... \mu_{k-1} \,1 0\, \mu_{k+1} ... \mu_n}\nonumber\\
	&= 2 \!\!\! \sum_{s_1,...,\widehat{s}_k,...,s_n } \!\!\!\!\!\!\!\!\! (-s_1)^{\mu_1} \cdots \widehat{(-s_k)}^{\mu_k} \cdots  (-s_n)^{\mu_n} \big( -|\psi^{(n+1)}_{s_1 ... s_{k-1} -+ s_{k+1} ... s_n}|^2 +  |\psi^{(n+1)}_{s_1 ... s_{k-1} +- s_{k+1} ... s_n}|^2 \big),
	\label{eq:localcurrentcalc0}
\end{align}
where $\widehat{(\cdot)}$ denotes omission. On the other hand, we compute:
\begin{align}
	(-1)^k \partial_{k,\mu_k} j^{\mu_1... \mu_k ... \mu_n}(x_1,...,x_n) \stackrel{\eqref{eq:currentincomponents}}{=}   &\sum_{s_1, ... ,s_n} \!\!\! (-s_1)^{\mu_1} \cdots (\widehat{-s_k)}^{\mu_k} \cdots  (-s_n)^{\mu_n}\nonumber\\
&~~~\times (-1)^k (\partial_{t_k} - s_k \partial_{z_k}) |\psi^{(n)}_{s_1...s_n}|^2(x_1,...,x_n).
\label{eq:localcurrentcalc1}
\end{align}
Using \eqref{eq:multitimeeqwithsources}, we have:
\begin{align}
	&(-1)^k(\partial_{t_k} - s_k \partial_{z_k}) |\psi^{(n)}_{s_1...s_n}|^2(x_1,...,x_n) \nonumber\\
	&= -i \big(\psi^{(n)}\big)^*_{s_1...s_n}(x_1,...,x_n)  \left( \sum_{t,u}A^{tu}_{s_k} \psi^{(n+1)}_{s_1...s_{k-1} \, t\, u \, s_{k+1}...s_n}(x_1,...,x_{k-1},x_k,x_k, x_{k+1},...,x_n) \right) + {\rm c.c.}\nonumber\\
	&= 2 \Im  \left(\big(\psi^{(n)}\big)^*_{s_1...s_n}(x_1,...,x_n)   \sum_{t,u} A^{tu}_{s_k} \psi^{(n+1)}_{s_1...s_{k-1} \, t\, u \, s_{k+1}...s_n}(x_1,...,x_{k-1},x_k,x_k, x_{k+1},...,x_n) \right).
	\label{eq:localcurrentcalc2}
\end{align}
Here ``c.c.'' denotes the complex conjugate of the previous summand.

Now we compare \eqref{eq:localcurrentcalc1} with \eqref{eq:localcurrentcalc0}. Demanding that both expressions be equal (condition \eqref{eq:currentcondition}) yields 
\begin{align}
	0 ~= ~2\!\!\! &\sum_{s_1, ...,\widehat{s}_k,... ,s_n} \!\!\!\!\!\!  (-s_1)^{\mu_1} \cdots (\widehat{-s_k)}^{\mu_k} \!\!\! \cdots  (-s_n)^{\mu_n} \left( \big(|\psi^{(n+1)}_{s_1 ... s_{k-1} -+ s_{k+1} ... s_n}|^2 \right.\nonumber\\
	& ~~~~~~~~~-  |\psi^{(n+1)}_{s_1 ... s_{k-1} +- s_{k+1} ... s_n}|^2\big)(x_1,...,x_{k-1},x_k,x_k,x_{k+1},...,x_n)\nonumber\\
	 & ~~~~~~~~~+ \tfrac{1}{2} (-1)^k \sum_{s_k} \left. (\partial_{t_k} - s_k \partial_{z_k}) |\psi^{(n)}_{s_1...s_n}|^2(x_1,...,x_n)\right).
\end{align}
This condition is certainly satisfied if all the summands vanish individually. Demanding that this be so and considering \eqref{eq:localcurrentcalc2}  gives the condition
\begin{align}
	&\left(-|\psi^{(n+1)}_{s_1 ... s_{k-1} -+ s_{k+1} ... s_n}|^2 +  |\psi^{(n+1)}_{s_1 ... s_{k-1} +- s_{k+1} ... s_n}|^2\right)(x_1,...,x_{k-1},x_k,x_k,x_{k+1},...,x_n) ~=
	\nonumber\\
	&  \sum_{s_k} \Im \left(\big(\psi^{(n)}\big)^*_{s_1...s_n}(x_1,...,x_n)   \sum_{t,u} A^{tu}_{s_k} \psi^{(n+1)}_{s_1...s_{k-1} \, t\, u \, s_{k+1}...s_n}(x_1,...,x_{k-1},x_k,x_k, x_{k+1},...,x_n) \right).
	\label{eq:localdetailedcond}
\end{align}
Now, except for additional spin indices and spacetime variables this condition has exactly the form of the condition for the model with only two sectors (see \eqref{eq:2currentbalance}). There the condition
\be
	\left( - |\psi^{(2)}_{-+}|^2 + |\psi^{(2)}_{+-}|^2 \right) (x,x) = \sum_{s_1} \Im \left(  \big(\psi^{(1)} \big)^*_{s_1}(x) \sum_{t,u} A^{tu}_{s_1} \psi^{(2)}_{tu}(x,x) \right)
\ee
was ensured by the IBC \eqref{eq:2modelIBC}
\be
	\psi^{(2)}_{-+} - e^{i\theta} \psi^{(2)}_{+-}(x,x) = \sum_s B^s \psi^{(1)}_s(x).
\ee
In an analogous way, we conclude that the IBCs \eqref{eq:mgeneralibcs} ensure the conditions \eqref{eq:localdetailedcond} for $k=1,...,n$; $n=1,...,N-1$. \qed

\subsection{Proof of theorem \ref{thm:massive}} \label{sec:proofmassive}

The proof of our main theorem built on the construction of a fixed point map $F$ in Sec. \ref{sec:proofmmain}. We now define an analogous map $\widetilde{F}$ with additional mass terms.

\begin{definition}[Definition (fixed point map in the massive case).]
Recall the definition of the map $F$ from the previous subsection. The massive fixed point map $\widetilde{F}: \mathcal{D} \to \mathcal{D}$ is defined in the same way, with the only difference that we use as the inhomogeneities 
\begin{equation}
	\widetilde{f}^{(n)}_{k} ~=~ f^{(n)}_k + m \gamma_k^0 v^{(n)}.
\end{equation}
\end{definition}

The well-definedness of $\widetilde{F}$ can be seen analogously as in the case $m = 0$. The statement of our theorem \ref{thm:massive} now comes from the combination of the two points in the following lemma.

\begin{lemma}
The map $\widetilde{F}$ has the following properties:
\begin{enumerate}
	\item[(i)] For every $T>0$, $\widetilde{F}$ has a unique fixed point in $\mathscr{D}$.
	\item[(ii)] $(\psii,...,\psi^{(N)}) \in \mathscr{B}_1 \oplus \dots \oplus \mathscr{B}_N$ is a $C^1_b$-solution of the IBC system \eqref{eq:multitimeeqwithsources},\eqref{eq:mgeneralibcs}, \eqref{eq:sourcetermswithmass}, \eqref{eq:mthsourcetermwithmass} with initial values given by \eqref{eq:minitial} if and only if it is in $\mathcal{D}$ and a fixed point of $\widetilde{F}$.
\end{enumerate}
\end{lemma}

\begin{proof}
\emph{Item (i).} The proof works like in the massless case. The only difference is that an additional summand
\begin{equation}
	\frac{m}{\gamma} \sum_{k=1}^n \left\| v^{(k)} - w^{(k)} \right\|_{\mathscr{B}_k, \gamma}
\end{equation}
appears on the right hand side of \eqref{eq:boundforfixedpoint}. For this, we have the estimate
\begin{equation}
	\frac{m}{\gamma}\sum_{k=1}^n \left\| v^{(k)} - w^{(k)} \right\|_{\mathscr{B}_k, \gamma} ~\leq~ \frac{m N}{\gamma} \left\| (v^{(1)},...,v^{(N)}) - (w^{(1)},...,w^{(N)}) \right\|_{\Banach_1 \oplus \dots \oplus \Banach_N, \gamma}.
\end{equation}

 Hence,
\begin{equation}\begin{split}
	& \left\| \widetilde{F}(v^{(1)},...,v^{(N)}) - \widetilde{F}(w^{(1)},...,w^{(N)}) \right\|_{\Banach_1 \oplus \dots \oplus \Banach_N, \gamma}\\ 
\leq~~ &\frac{1}{\gamma} 
(2N)^{4N+1} \left\| A \right\|_\infty \max \{1, \left\| B \right\|_{\infty}^{N} \} \left\| (v^{(1)},...,v^{(N)}) - (w^{(1)},...,w^{(N)}) \right\|_{\Banach_1 \oplus \dots \oplus \Banach_N, \gamma}
\\ +~ &N^2 \frac{m}{\gamma} \left\| (v^{(1)},...,v^{(N)}) - (w^{(1)},...,w^{(N)}) \right\|_{\Banach_1 \oplus \dots \oplus \Banach_N, \gamma}.
\end{split}
\end{equation}
Also in this case, $\gamma$ can be chosen large enough to make $\widetilde{F}$ a contraction, thus Banach's fixed point theorem yields the claim.
\\ 

\noindent\emph{Item (ii).} The proof is exactly analogous to the one in the massless case, noting that the mass term does not obstruct the consistency condition \eqref{eq:consistencymsectors}. \qed
\end{proof}

\section{Discussion} \label{sec:discussion}

By constructing an explicit model of Dirac particles in 1+1 dimensions, we have shown that interior-boundary conditions can be combined with Dirac's concept of multi-time wave functions. As the IBC approach uses the particle-position representation,  multi-time wave functions are necessary to make the IBC approach manifestly covariant. We have then demonstrated rigorously the existence and uniqueness of solutions and identified a wide class of IBCs which ensure probability conservation on arbitrary Cauchy surfaces, in particular in all Lorentz frames. Overall, we have obtained a new rigorous and (almost fully) covariant model QFT.
As discussed in Sec.\ \ref{sec:li}, ``almost fully covariant'' means that the model contains certain constant matrices $A, B$ which would have to transform in a certain way for the model to be Lorentz invariant. We believe that this li\-mi\-ta\-tion is due to the fact that the model uses only fermions and does not involve exchange particles (e.g.\ photons) yet.
 As more sophisticated models could overcome this point, we think that our work demonstrates sufficiently clearly that IBCs can be combined with relativity.

Note that there are rigorous, interacting and Lorentz invariant QFT models for Dirac particles in 1+1 dimensions, for example the Thirring model \cite{thirring}. The Thirring model achieves Lorentz invariance through a quartic interaction term in the Lagrangian, \newline $(\overline{\psi} \gamma_\mu \psi) (\overline{\psi} \gamma^\mu \psi)$, while our model corresponds to a cubic term in the field operator $\psi$. While the former is necessary in 1+1 dimensions to achieve Lorentz invariance, the advantage of the latter is that it is closer to the interaction term of QED, $A_\mu \, \overline{\psi} \gamma^\mu \psi$, which also contains three field operators. In this regard, our model comes closer to the realistic situation (which, besides simplicity and tractability, has been our main criterion for setting up a toy model).

While we have explained how to obtain a multi-time formulation of the IBC approach only at the example of a particular model, we think that some of the basic insights can also be transferred to more general cases. For higher space-time dimensions, for example, we expect that a similar consideration as used in Sec.\ \ref{sec:prooflocalcurrentcons} to obtain the condition \eqref{eq:localcurrentcons} for local probability conservation will lead to the correct IBCs. This consideration only requires that the respective free multi-time equations have a conserved tensor current $j^{\mu_1 ... \mu_n}(x_1,...,x_n)$ with the right properties for each sector of Fock space.

For higher dimensions, though, one will have to take into account that the codimension of the set of coincidence points $\coincidence$ is greater than for 1+1 dimensions. To determine the probability flow into $\coincidence$, one then has to take integrals over balls with infinitesimal radii instead of mere limits towards $\coincidence$. This will lead to the appearance of integrals also in the IBC, as well as in the term which then replaces $A \psi^{(n+1)}$ in the equation of motion for $\psi^{(n)}$. The integrals in the IBC will then force the wave function to become divergent, i.e., more singular than in the 1+1 dimensional case.  It seems that such a strong singularity is related to the ultraviolet divergence problem. It would be a highly interesting question for future research to see whether the IBC approach can also help to overcome the UV problem in these cases. Previous works in the non-relativistic case \cite{ibc1,ibc2,ls:2018,lampart:2018} suggest that this may be so; however, in the relativistic case (and for the Dirac Hamiltonian) the question is completely open. An important obstacle is that point interactions do not seem to be possible for the Dirac operator in 1+3 dimensions \cite{Sve81}.

As indicated before, it would be desirable to treat more realistic QFTs involving exchange bosons with the multi-time IBC approach in the future. The exchange particles would help to make the theory fully Lorentz invariant. However, one then requires a suitable particle-position representation of bosonic wave functions which is an open problem by itself. For photons, the difficulties have recently been discussed in \cite{ktz:2018} and a new wave equation for the photon has been suggested. It might be possible to combine a multi-particle (and multi-time) version of that equation with our model (see \cite{kltz:2019} which already achieves contact interactions for a photon-electron system). As the probability current of \cite{ktz:2018} involves a preferred time-like vector field, the IBCs (and consequently the whole wave function dynamics) would then likely depend on that preferred vector field.

\vspace{0.4cm}

\noindent \textbf{Acknowledgments.}\\[1mm]
We would like to thank Detlef D\"urr, Stefan Keppeler, Jonas Lampart, Julian Schmidt and Stefan Teufel for helpful discussions. Special thanks go to Roderich Tumulka for advice and helpful comments on the manuscript. We would furthermore like to thank an anonymous referee for helpful suggestions. L. N. gratefully acknowledges financial support by the Studienstiftung des deutschen Volkes.\\[1mm]
\begin{minipage}{15mm}
\includegraphics[width=13mm]{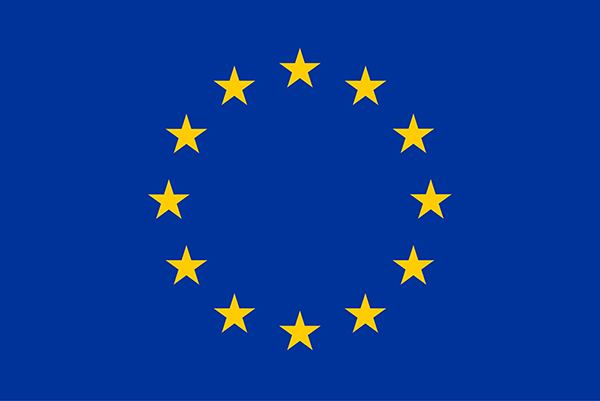}
\end{minipage}
\begin{minipage}{143mm}
This project has received funding from the European Union's Framework for
Re- \\ search and Innovation Horizon 2020 (2014--2020) under the Marie Sk{\l}odowska-
\end{minipage}\\[1mm]
Curie Grant Agreement No.~705295.


\end{document}